\documentclass[a4paper,UKenglish]{lmcs}
\pdfoutput=1

\usepackage{lastpage}
\lmcsdoi{17}{3}{14}
\lmcsheading{}{\pageref{LastPage}}{}{}%
{May~25,~2020}{Jul.~30,~2021}{}

\usepackage[utf8]{inputenc}


\usepackage{amsmath,amssymb}
\usepackage{mathtools}  % for \coloneqq
\usepackage{xspace}

\usepackage{enumitem}
\setlist[itemize, 1]{label={\tiny$\bullet$}}

\usepackage{tikz}
\usetikzlibrary{calc, positioning}

\usepackage{multicol}

\newcommand{\quotes}[1]{``#1''}
\newcommand{\ptime}{\textsc{Ptime}\xspace}
\newcommand{\np}{\textsc{NP}\xspace}
\newcommand{\npc}{\textsc{NP-complete}\xspace}
\renewcommand{\models}{\vDash}
\newcommand{\nmodels}{\nvDash}
\renewcommand{\emptyset}{\varnothing}
\newcommand{\ceq}{\coloneqq}
\newcommand{\eqc}{\eqqcolon}
\newcommand{\subf}{\operatorname{subf}}
\newcommand{\free}{\operatorname{free}}
\newcommand{\dom}{\operatorname{dom}}
\newcommand{\nh}[2]{\operatorname{N}_{#1}(#2)} % neighbourhood
\newcommand{\rsim}{_{/\sim}}
\newcommand{\target}{\mathcal{T}}
\newcommand{\res}[1]{{\uhr_{#1}}}
\newcommand{\arity}[1]{\operatorname{ar}(#1)}

\newcommand{\vark}{{\mathcal{V}_k}}

\newcommand{\fo}{\ensuremath{\operatorname{FO}}}
\newcommand{\eso}{\ensuremath{\operatorname{\Sigma^1_1}}}
\newcommand{\exclogic}{\ensuremath{\fo({}\excl{})}}
\newcommand{\deplog}{\ensuremath{\fo(\operatorname{dep})}}
\newcommand{\inclogic}{\ensuremath{\fo(\incl)}}
\newcommand{\inexlogic}{\fo(\incl,{}\excl{})}
\newcommand{\gfplogic}{\ensuremath{\text{GFP}^+}}
\newcommand{\union}{\mathcal{U}}

\newcommand{\dep}{\operatorname{=}}
\newcommand{\indep}{\bot}
\newcommand{\incl}{\subseteq}
\newcommand{\excl}{\mathbin{|}}
\newcommand{\Eincl}{\ensuremath{{E_{\mathsf{in}}}}}
\newcommand{\Eexcl}{\ensuremath{{E_{\mathsf{ex}}}}}
\newcommand{\init}{\ensuremath{I}}
\newcommand{\win}{\cup\mathsf{-game}}

\newcommand{\gcomp}[1]{{\cG^\vartriangle_{#1}}}
\newcommand{\secorder}{second-order\xspace}
\newcommand{\erich}{Gr\"adel\xspace}
\newcommand{\jouko}{V\"a\"an\"anen\xspace}

\renewcommand{\phi}{\varphi}
\renewcommand{\theta}{\vartheta}
\renewcommand{\epsilon}{\varepsilon}
\newcommand{\phiwin}{\phi_{\operatorname{win}}}
\newcommand{\phitarget}{\phi_{\target}}
\newcommand{\psitarget}{\psi_{\target}}
\newcommand{\psiwin}{\psi_{\operatorname{win}}}
\newcommand{\psimove}{\psi_{\operatorname{move}}}
\newcommand{\psiexcl}{\psi_{\Eexcl}}
\newcommand{\psiinit}{\psi_{\operatorname{init}}}
\newcommand{\defiff}{\mathbin{:\!\!\iff}}
\newcommand{\A}{\forall}
\newcommand{\E}{\exists}
\newcommand{\ra}{\rightarrow}
\newcommand{\lra}{\leftrightarrow}
\newcommand{\uhr}{\upharpoonright}
\newcommand{\cupdot}{\mathbin{\mathaccent\cdot\cup}}

\newcommand{\bN}{\mathbb{N}}

\newcommand{\cG}{\mathcal{G}}
\newcommand{\cI}{\mathcal{I}}
\newcommand{\cS}{\mathcal{S}}

\newcommand{\fA}{\mathfrak{A}}
\newcommand{\fB}{\mathfrak{B}}
\newcommand{\tuple}[1]{{\bar{#1}}}
\newcommand{\ta}{\tuple{a}}
\newcommand{\tb}{\tuple{b}}
\newcommand{\tc}{\tuple{c}}
\newcommand{\tf}{\tuple{f}}

\newcommand{\ts}{\tuple{s}}
\renewcommand{\tt}{\tuple{t}}
\newcommand{\tu}{\tuple{u}}
\newcommand{\tv}{\tuple{v}}
\newcommand{\tw}{\tuple{w}}
\newcommand{\tx}{\tuple{x}}
\newcommand{\ty}{\tuple{y}}
\newcommand{\tz}{\tuple{z}}

\newcommand{\tP}{\tuple{P}}
\newcommand{\tR}{\tuple{R}}

\newcommand{\tepsilon}{\tuple{\epsilon}}
\newcommand{\pot}[1]{\ensuremath{\mathcal{P}(#1)}}
\newcommand{\potne}[1]{\ensuremath{\mathcal{P}^+(#1)}}
\newcommand{\set}[2]{\ensuremath{\{ #1 : #2\}}}
\newcommand{\mc}[2]{\ensuremath{\cG(#1, #2)}}
\newcommand{\mcX}[2]{\ensuremath{\cG_X(#1, #2)}}

\begin{document}

\title[On the Union Closed Fragment]{On the Union Closed Fragment of Existential Second-Order Logic and Logics with Team Semantics}

\author[M.~Hoelzel]{Matthias Hoelzel}

\author[R.~Wilke]{Richard Wilke}
% \address{Mathematical Foundations of Computer Science, RWTH Aachen University, Aachen, Germany}	%optional
% \email{wilke@logic.rwth-aachen.de}  %optional

\address{Mathematical Foundations of Computer Science, RWTH Aachen University, Aachen, Germany}	%required
\email{hoelzel@logic.rwth-aachen.de, wilke@logic.rwth-aachen.de}  %optional

\thanks{The research of Matthias Hoelzel has been supported by the German
Research Foundation (DFG). Richard Wilke is a doctoral student in the Research Training Group 2236 UnRAVeL, also funded by DFG. Both authors carried out this research as part of their PhD studies, under the supervision of Erich \erich.}
%\thanks{The second author is supported by the DFG, Research Training Group 2236 UnRAVeL}	%optional

\maketitle

\begin{abstract}
    We present syntactic characterisations for the union closed fragments of existential \secorder logic and of logics with team semantics.
    Since union closure is a semantical and undecidable property, the normal form we introduce enables the handling and provides a better understanding of this fragment.
    We also introduce inclusion-exclusion games that turn out to be precisely the corresponding model-checking games.
    These games are not only interesting in their own right, but they also are a key factor towards building a bridge between
    the semantic and syntactic fragments.
    On the level of logics with team semantics we additionally present restrictions of inclusion-exclusion logic to capture the union closed fragment.
    Moreover, we define a team based atom that when adding it to first-order logic also precisely
    captures the union closed fragment of existential \secorder logic which answers an open question by Galliani and Hella.
\end{abstract}

\section{Introduction}

One branch of model theory engages with the characterisation of semantical fragments, which typically are undecidable, as syntactical fragments of the logics under consideration.
Prominent examples are van Benthem's Theorem characterising the bisimulation invariant fragment of first-order logic as the modal-logic \cite{Ben76} or preservation theorems like the {\L}o\'{s}-Tarski Theorem, which states that formulae preserved in substructures are equivalent to universal formulae \cite{HodBook}.
In this paper we consider formulae $\phi(X)$ of existential second-order logic, $\eso$, in a free relational variable $X$ and investigate the property of being closed under unions,
meaning that whenever a family of relations $X_i$ all satisfy $\phi$, then their union $\bigcup_i X_i $ should also do so.
Certainly closure under unions is an undecidable property.
We provide a syntactical characterisation of all formulae of existential second-order logic obeying this property via a normal form called \emph{myopic}-$\eso$, a notion based on ideas of Galliani and Hella \cite{GalHel13}.
By Fagin's Theorem, $\eso$ is the logical equivalent of the complexity class \np which highlights the importance to understand its fragments.
Towards this end we employ game theoretic concepts and introduce a novel game type, called \emph{inclusion-exclusion games}, suited for formulae $\phi(X)$ with a free relational variable.
In these games a strategy no longer is simply winning for one player --- and hence proving whether a sentence is satisfied --- but it is moreover adequate for a certain relation $Y$ over $\fA$ showing that the formula is satisfied by $\fA$ and $Y$, in symbols $\fA\models\phi(Y)$.
We construct myopic-$\eso$ formulae that can define the winning regions of specifically those inclusion-exclusion games that are (semantically) closed under unions.
Conceptually such games are eligible for any $\eso$-formula, but since our interest lies in those formulae that are closed under unions, we introduce a restricted version of such games, called \emph{union games}, that precisely correspond to the model-checking games of union closed $\eso$-formulae.
Consequently, the notion of union closure is captured on the level of formulae by the myopic fragment of $\eso$ and on the game theoretic level by union games.

Existential second-order logic has a tight connection to modern logics of dependence and independence that are based on the concept of teams, introduced by Hodges \cite{Hod97}, and later refined by \jouko in 2007 \cite{Vaa07}.
In contrast to classical logics, formulae of such a logic are evaluated against a set of assignments, called a \emph{team}.
One main characteristic of these logics is that dependencies between variables, such as ``$x$ depends solely on $y$'', are expressed as atomic properties of teams.
Widely used dependency atoms include dependence ($\dep(x,y)$), inclusion ($x\incl y$), exclusion ($x\excl y$) and independence ($x\indep y$).
It is known that both independence logic $\fo(\indep)$ and inclusion-exclusion logic $\inexlogic$ have the same expressive power as full existential second-order logic $\eso$ \cite{Gal12}.
The team in such logics corresponds to the free relational variable in existential second-order formulae, enabling us to ask the same questions about fragments with certain closure properties in both frameworks.
One example of a well understood closure property is \emph{downwards closure} stating that if a formula is satisfied by a team then it is also satisfied by all subteams (i.e.~subsets of that team).
It is well known that exclusion logic $\exclogic$ is equivalent to dependence logic $\deplog$ \cite{Gal12},
which corresponds to the downwards closed fragment of $\eso$ \cite{KonVaa09}.
The issue of union closure is different.
Galliani and Hella have shown that inclusion logic $\inclogic$ corresponds to greatest fixed-point-logic $\gfplogic$ and, hence, by using the Immerman-Vardi Theorem, it captures all $\ptime$ computable queries on ordered structures \cite{GalHel13}.
They also proved that every union closed dependency notion that itself is first-order definable (where the formula has access to a predicate for the team) is already definable in inclusion logic.
However, there are union closed properties that are not definable in inclusion logic (think of a union closed $\np$ property).
For a concrete example we refer to the atom $\mathcal{R}$ from \cite{GalHel13}.
Thus Galliani and Hella asked the question whether there is a union closed atomic dependency notion $\beta$, such that the logic $\fo(\beta)$ captures precisely the union closed fragment of $\inexlogic$.
In the present work we answer this question positively with the aid of inclusion-exclusion games.
Furthermore, we present a syntactical restriction of all $\inexlogic$ formulae that also precisely describe the union closed fragment.
This syntactical fragment corresponds to myopic-$\eso$ and is in harmony with the game theoretical view, which is described by union games.

This paper is based on \cite{HW19, HW20} and also contains some improvements from \cite{HoelzelPhD} like the more direct translations into the myopic fragments (Theorem \ref{thm: eso-sentence => myopic} and Theorem \ref{thm: myopic inex - more general}), while Theorem \ref{thm: union closed eso => myopic stronger} and Theorem \ref{thm: union closed inex => myopic inex stronger} are based on the original proofs and are enhanced by the observation that they produce only a limited number of in-/exclusion atoms or literals with second-order symbols.
Furthermore, we also present restricted variants of the inclusion-exclusion games that are the model-checking games of $\inclogic$ or $\exclogic$.

Sections \ref{sec: in/ex games}, \ref{sec:char U in eso} and \ref{sec:union games}
deal with second-order logic and can be read without knowledge about team semantics.
In section three the central notion of this paper, inclusion-exclusion
games, is introduced, which is used in section four to characterise
the union closed fragment within existential second-order logic.
Section five provides a restriction of the games specifically suited for this fragment.
Based on these notions, Section \ref{sec:myopic inex} describes the union closed
fragment of inclusion-exclusion logic in terms of syntactical restrictions.
The question of Galliani and Hella, whether there is a union closed atom that
constitutes the union closed fragment, is answered positively in Section \ref{sec: atom}.
Finally, in Section \ref{sec: in and ex games} we discuss restrictions of
inclusion-exclusion games suited for example for the downwards closed fragment.

\section{Preliminaries}

We assume familiarity with first-order logic and existential second-order logic, $\fo$ and $\eso$ for short.
For a background we refer to the textbook \cite{GraKLMSVVW07}.

\subsection*{Graphs}

Let $G = (V, E)$ be a graph, that is $V$ is a possibly empty set of vertices and $E \subseteq V\times V$ the edge set.
The neighbourhood of a vertex $v$ in $G$ is denoted by $\nh{G}{v} = \{w \in V : (v,w) \in E\}$.
The restriction of $G$ to $W\subseteq V$ is defined as $G\res{W} \ceq (W, E\cap (W\times W))$.
By $G-W$ we denote the graph $G\res{V\setminus W}$.
For a set $F\subseteq V\times V$ the extension of $G$ by $F$ is denoted by $G + F \ceq (V, E\cup F)$.

\subsection*{Logic}

For a given $\tau$-structure $\fA$ with universe $A$\footnote{We write structures in Fraktur letters and use the corresponding Latin letters to denote their respective universe.} and formula $\phi(\tx)$ we define $\phi^\fA \ceq \set{\ta \in A^{|\tx|}}{\fA \models \phi(\ta)}$, $\free(\phi)$ is the set of free first-order variables and $\subf(\psi)$ is the set of subformulae of $\psi$; sometimes for technical reasons it is necessary to consider $\subf(\psi)$ as a \emph{multiset}, and we will do so tacitly.
Notations like $\tv, \tw$ always indicate that $\tv = (v_1, \dotsc, v_k)$ and $\tw = (w_1, \dotsc, w_\ell)$ are some (finite) tuples.
Here $k = |\tv|$ and $\ell = |\tw|$, so $\tv$ is a $k$-tuple while $\tw$ is an $\ell$-tuple.
We write $\{\tv\}$ or $\{ \tv, \tw \}$ as abbreviations for $\{ v_1, \dotsc, v_k \}$ resp.~$\{ v_1, \dotsc, v_k, w_1, \dotsc, w_\ell \}$
while $\{ (\tv), (\tw) \}$ is the set consisting of the two tuples $\tv$ and $\tw$ (as elements).
The \emph{concatenation} of $\tv$ and $\tw$ is $(\tv, \tw) \ceq (v_1, \dotsc, v_k, w_1, \dotsc, w_\ell)$.
The power set of a set $A$ is denoted by $\pot{A}$ and $\potne{A} \ceq \pot{A}\setminus\{\emptyset\}$.

\subsection*{Team Semantics}

A \emph{team} $X$ over $\fA$ is a \emph{set} of assignments mapping a common domain $\dom(X) = \{ \tx \}$ of variables
into $A$.
The restriction of $X$ to some first-order formula $\phi(\tx)$ is $X \res{\phi} \ceq \set{s \in X}{\fA \models_s \phi}$.
For a given subtuple $\ty = (y_1, \dots, y_\ell)\subseteq \tx$ and every $s \in X$ we define $s(\ty) \ceq (s(y_1), \dotsc, s(y_\ell))$.
Furthermore, we frequently use $X(\ty) \ceq \set{s(\ty)}{s \in X}$, which is an $\ell$-ary relation over $\fA$.
For an assignment $s$, a variable $x$ and $a \in A$ we use $s[x \mapsto a]$ to denote the assignment resulting from $s$
by adding $x$ to its domain (if it is not already contained) and declaring $a$ as the image of $x$.

\begin{defi}
    \label{def: fo ts}
    Let $\fA$ be a $\tau$-structure, $X$ a team of $\fA$.
    In the following $\lambda$ denotes a first-order $\tau$-literal and $\phi, \psi$ arbitrary formulae in negation normal form.
    \begin{itemize}
        \item $\fA\models_X \lambda \defiff \fA\models_s\lambda$ for all $s\in X$
        \item $\fA\models_X \phi\land\psi \defiff \fA\models_X\phi$ and $\fA\models_X\psi$
        \item $\fA\models_X \phi\lor\psi \defiff \fA\models_Y\phi$ and $\fA\models_Z\psi$ for some $Y, Z \subseteq X$ such that $Y \cup Z = X$
        \item $\fA\models_X \A x\phi \defiff \fA\models_{X[x\mapsto A]}\phi$
        \item $\fA\models_X \E x\phi \defiff \fA\models_{X[x\mapsto F]}\phi$ for some $F\colon X\to \potne{A}$
    \end{itemize}
    Here $X[x \mapsto A] \ceq \set{s[x \mapsto a]}{s \in X,\, a \in A}$ and $X[x \mapsto F] \ceq \set{s[x \mapsto a]}{s \in X,\,a \in F(s)}$.
\end{defi}

\noindent
Team semantics for a first-order formula $\phi$ (without any dependency concepts) boils down to evaluating $\phi$
against every single assignment, i.e.~more formally we have $\fA \models_X \phi \iff \fA \models_s \phi$  for \emph{every} $s \in X$ (in usual Tarski semantics).
This is also known as the \emph{flatness property} of $\fo$.
The reason for considering teams instead of single assignments is that they allow the formalisation of dependency statements in the form of dependency atoms.
Among the most frequently considered atoms are the following.

\begin{description}
    \item[Dependence] $\fA\models_X\dep(\tx,y) \defiff s(\tx)=s'(\tx)\text{ implies }s(y)=s'(y)\text{ for all }s,s'\in X$
    \item[Inclusion] $\fA\models_X\tx\incl\ty \defiff X(\tx)\subseteq X(\ty)$
    \item[Exclusion] $\fA\models_X\tx\excl\ty \defiff X(\tx) \cap X(\ty) = \emptyset$
    \item[Independence] $\fA\models_X\tx\indep\ty \defiff \text{for all }s,s'\in X\text{ exists }s''\in X\text{ s.t.~}s(\tx)=s''(\tx)$ and $s'(\ty)=s''(\ty)$
\end{description}

\noindent
Dependence was introduced by \jouko \cite{Vaa07}, inclusion and exclusion by Galliani \cite{Gal12} and independence by \erich and \jouko \cite{GraVaa13}.
When we speak about first-order team logic augmented with a certain atomic dependency notion, for example inclusion, we denote it by writing $\fo(\incl)$ and so forth.
All of these logics have the empty team property, which means that $\fA \models_\emptyset \phi$ is always true.
This is also the reason why \emph{sentences} are not evaluated against $\emptyset$
but rather against $\{ \emptyset \}$, which is the team consisting of the empty assignment.
Let $\phi$ be a first-order formula and $\psi$ be any formula of a logic with team semantics.
We define $\phi \ra \psi$ as $\operatorname{nnf}(\neg\phi) \lor (\phi \land \psi)$ where $\operatorname{nnf}(\neg\phi)$ is the negation normal form of $\neg \phi$.
It is easy to see that $\fA \models_X \phi \ra \psi \iff \fA \models_{X \res{\phi}} \psi$.

\subsubsection*{Union Closure}

A formula $\phi$ of a logic with team semantics is said to be \emph{union closed} if whenever $\fA\models_{X_i}\phi$ holds for all $i\in I$ then also $\fA\models_X\phi$, where $X = \bigcup_{i\in I}X_i$.
Analogously, a formula $\phi(X)$ of $\eso$ with a free relational variable $X$ is union closed if $\fA\models\phi(X_i)$ for all $i\in I$ implies $\fA\models\phi(X)$.

\subsection*{FO Interpretations}

A first-order interpretation from $\sigma$ to $\tau$ of arity $k$ is a sequence
$\cI = (\delta, \epsilon,  (\psi_S)_{S \in \tau})$ of $\fo(\sigma)$-formulae, called the domain, equality and relation formulae respectively.
We say that $\cI$ interprets a $\tau$-structure $\fB$ in some $\sigma$-structure $\fA$ and write $\fB \cong \cI(\fA)$ if and only if
there exists a surjective function $h$, called the coordinate map, that maps $\delta^\fA = \set{\ta \in A^k}{\fA \models \delta(\ta)}$ to $B$ preserving and reflecting the equalities and relations provided by $\epsilon$ and $\psi_S$, such that $h$ induces an isomorphism between the quotient structure $(\delta^\fA,(\psi^\fA_S)_{S \in \tau })/\epsilon^\fA$ and $\fB$.
A more detailed explanation can be found in \cite{GraKLMSVVW07}.
For a $\tau$-formula $\phi$ we associate the $\sigma$-formula $\phi^\cI$ by relativising quantifiers to $\delta$, using $\epsilon$ as equality and $\psi_S$ instead of $S$.
We extend this translation to $\eso$ by the following rules for additional free/quantified relation symbols $S$.
\begin{itemize}
    \item $(\E S \vartheta)^\cI \ceq \E S^\star \big( \forall \tx_1 \dotsb \tx_{\arity{S}} \big(S^\star \tx_1 \dotsb \tx_{\arity{S}} \ra \bigwedge^{\arity{{S}}}_{j=1} \delta(\tx_j) \big) \land \vartheta^\cI\big)$,
    \item $(S v_1\dotsb v_{\arity{S}})^\cI \ceq \exists \tw_1 \dotsb \tw_{\arity{S}} \big(\bigwedge^{\arity{S}}_{j=1} (\delta(\tw_j) \land \epsilon(\tv_j, \tw_j)) \land S^\star \tw_1\dotsb \tw_{\arity{S}}\big)$.
\end{itemize}
An assignment $s \colon \{ \tx_1, \dots, \tx_m \} \to A$ is well-formed (w.r.t.~$\cI$), if $s(\tx_i) \in \delta^\fA (= \dom(h))$ for every $i=1,\dots,m$.
Such an assignment encodes $h \circ s \colon \{x_1, \dots, x_m\} \to B$ with $(h \circ s)(x_i) \ceq h(s(\tx_i))$ which is an assignment over $\fB$.
Similarly, a relation $Q$ is well-formed (w.r.t.~$\cI$), if $Q \subseteq (\delta^\fA)^\ell$ where $\ell = \frac{\arity{Q}}{k} \in \bN$,
and we define $h(Q) \ceq \set{(h(\ta_1), \dots, h(\ta_\ell))}{(\ta_1, \dots, \ta_\ell) \in Q}$, which is the $\ell$-ary relation over $\fB$ that was described by $Q$.
The connection between $\phi^\cI$ and $\phi$ is made precise in the well-known interpretation lemma.

\begin{lem}[Interpretation Lemma for \eso]
    \label{lem: interpretation lemma for sigma^1_1}
    Let $\phi(S_1,\dotsc,S_n) \in \eso$.
    Let $R^\star_i \subseteq A^{k \cdot \arity{S_i}}$ for all $i$ and
    $s \colon \{ \tx_1, \dots, \tx_m \} \to A$ be well-formed.
    Then:
    $(\fA, R_1^\star, \dots, R^\star_n) \models_s \phi^\cI \iff (\fB, h(R^\star_1), \dots, h(R^\star_n)) \models_{h \circ s} \phi.$
\end{lem}

\section{Inclusion-Exclusion Games}
\label{sec: in/ex games}

Classical model-checking games are designed to express satisfiability of sentences, i.e.~formulae without free variables.
The instantiation of a first-order formula with free variables by some values is dealt with by constructing expanded structures having constants that are assigned the respective values.
Since our focus lies on second-order formulae in a free relational variable we are in need for a game that is able to not only express that a formula is satisfied, but moreover that it is satisfied by a certain relation.
In the games we are about to describe a set of designated positions is present --- called the \emph{target set} --- which corresponds to the full relation $A^k$ (where the free relational variable has arity $k$).
A winning strategy is said to be \emph{adequate} for a subset $X$ of the target positions, if the target vertices visited by it are $X$.
On the level of logics this matches the relation satisfying the corresponding formula, i.e.~there is a winning strategy adequate for $X$ if and only if the formula is satisfied by $X$.

\begin{defi}
    \label{def: inex game}
    An \emph{inclusion-exclusion game} $\cG = (V, V_0, V_1, E, \init, T, \Eexcl)$ is
    played by two players 0 and 1 where
    \begin{itemize}
        \item $V_\sigma$ is the set of vertices of player $\sigma$,
        \item $V = V_0 \cupdot V_1$,
        \item $E \subseteq V \times V$ is a set of possible moves,
        \item $\init \subseteq V$ is the (possibly empty) set of initial positions,
        \item $T \subseteq V$ is the set of target vertices and
        \item $\Eexcl \subseteq V \times V$ is the exclusion condition, which defines the winning condition for player $0$.\footnote{$\Eexcl$ can always be replaced by the symmetric closure of $\Eexcl$ without altering its semantics.}
    \end{itemize}
    The edges going into $T$, that is $\Eincl\ceq E\cap(V\times T)$, are called \emph{inclusion edges}, while $\Eexcl$ is the set of \emph{exclusion edges} (sometimes also called conflicting pairs).
\end{defi}

\noindent
Inclusion-exclusion games are \emph{\secorder} games, so instead of single plays we are more interested in sets of plays that are admitted by some winning strategy for player $0$.
Model-checking games for logics with team semantics made their first appearance in \cite{Graedel13} and, in fact, inclusion-exclusion games can be considered to be a variant of the second-order games introduced in the same paper.

For a subset $X\subseteq T$ the aim of player 0 is to provide a winning strategy (which can be viewed as a set of plays respecting the exclusion condition and containing all possible strategies of player 1) such that the vertices of $T$ that are visited by this strategy correspond precisely to $X$.

\begin{defi}
    \label{def: winning strategy}
    A \emph{winning strategy} (for player 0) $\cS$ is a possibly empty subgraph $\cS = (W,F)$ of $G = (V,E)$ ensuring the following four consistency conditions.

\smallskip
    \noindent\begin{minipage}{\linewidth}
    \begin{enumerate}
        \item For every $v \in W \cap V_0$ holds $\nh{\cS}{v} \neq \emptyset$. \label{def: winning strategy - player 0}
        \item For every $v \in W \cap V_1$ holds $\nh{\cS}{v} = \nh{G}{v}$. \label{def: winning strategy - player 1}
        \item $\init \subseteq W$. \label{def: winning strategy - initial position}
        \item $(W \times W) \cap \Eexcl = \emptyset$. \label{def: winning strategy - exclusion edges}
    \end{enumerate}
    \end{minipage}
\end{defi}

\noindent
Intuitively, the conditions (\ref{def: winning strategy - player 0}) and (\ref{def: winning strategy - player 1})
state that the strategy must provide at least one move from each node of player $0$ used by the strategy but does not make assumptions about the moves that player $1$ may make whenever the strategy
plays a node belonging to player $1$.
In particular, the strategy must not play any terminal vertices that are in $V_0$.
Furthermore, (\ref{def: winning strategy - initial position}) enforces that at least the initial vertices are contained while
(\ref{def: winning strategy - exclusion edges}) disallows playing with conflicting pairs $(v,w) \in \Eexcl$, i.e.~$v$ and $w$ must not
coexist in any winning strategy for player~$0$.
If $\init = \emptyset$, then $(\emptyset, \emptyset)$ is the trivial winning strategy.
Another interesting observation is that $(W,F)$ is a winning strategy for player~$0$, if and only if the subgraph of the game induced by $W$, which is $(W,(W\times W) \cap E)$, is a winning strategy.
In other words, me may freely add edges that are present in the game graph between nodes of a winning strategy, i.e.\ a winning strategy is determined by its nodes.

Since we do \emph{not} have a notion for a winning strategy for player $1$, inclusion-exclusion games can be viewed as solitaire games.

Of course, the winning condition of an inclusion-exclusion game $\cG$ is first-order definable.
The formula $\phiwin(W)$ has the property that $\cG\models\phiwin(W)$ if and only if $(W,(W\times W) \cap E)$ is a winning strategy for player 0 in $\cG$, where
%\begin{align*}
%    \phiwin(W,F) \; \ceq \; & \A v ( Wv \ra ((V_0 v \land \E w (Evw \land Ww \land Fvw)) \,\lor \\
%                            & \phantom{\A v ( Wv \ra (} (V_1 v \land \A w (Evw \ra Ww \land Fvw))) ) \, \land \\
%    &     \A v (\init v \ra Wv) \land \A v \A w ((Wv \land Ww) \ra \neg \Eexcl vw)
%\end{align*}
\begin{align*}
    \phiwin(W) \; \ceq \; & \A v ( Wv \ra [(V_0 v \land \E w (Evw \land Ww)) \,\lor \\
    & \phantom{\A v [ Wv \ra (]} (V_1 v \land \A w (Evw \ra Ww))] ) \, \land  \\
    &  \A v (\init v \ra Wv) \land \A v \A w ((Wv \land Ww) \ra \neg \Eexcl vw).
\end{align*}
describes the winning conditions imposed on the subgraph induced by $W$.

We are mainly interested in the subset of target vertices that are visited by a winning strategy $\cS = (W,F)$.
More formally, $\cS$ induces $\target(\cS) \ceq W \cap T$, which we also call the \emph{target} of $\cS$.
This definition is reminiscent to the notion $\text{Team}(\cS, \psi)$ from \cite{Graedel13}, but instead of relying on a formula $\psi$ we are using the component $T$ of the game to define $\target(\cS)$.
Now we can associate with every inclusion-exclusion game $\cG$ the set of targets of winning strategies:
$\target(\cG) \ceq \set{\target(\cS)}{\cS \text{ is a winning strategy for player 0 in } \cG}$.

Intuitively, as already pointed out, games of this kind will also be the model-checking games for $\eso$-formulae $\phi(X)$ that have a free relational variable $X$.
Given a structure $\fA$ and such a formula, we are interested in the possible relations $Y$ that satisfy the formula, in symbols $(\fA, Y)\models\phi(X)$.
We will construct the game such that $Y$ satisfies $\phi$ if and only if there is a strategy of player 0 winning for the set $Y\subseteq T$, thus $\target(\cG) = \{Y : (\fA, Y)\models\phi\}$.
It will be more convenient for our purposes that the target vertices of an inclusion-exclusion game are not required to be terminal positions.
However it would be no restriction as it is easy to transform any given game into one that agrees on the (possible) targets, in which all target vertices are terminal.

Let us start the analysis of inclusion-exclusion games from a complexity theoretical point of view.
The satisfiability problem of propositional logic can be reduced to deciding whether player 0 has a winning strategy in an inclusion-exclusion game.

\begin{thm}
    \label{thm: winning inex games is npc}
    The problem of deciding whether $X\in\target(\cG)$ for a finite inclusion-exclusion game $\cG$ is \npc in the size of $\cG$.
\end{thm}
\begin{proof}
    Determining whether $X\in\target(\cG)$ holds is clearly in \np, as the winning strategy can be guessed and verified in polynomial time.

    For the $\np$-hardness we present a reduction from the satisfiability problem of propositional logic.
    Let $\phi$ be formula in conjunctive normal form, i.e.~$\phi = \bigwedge_{j\leq m} C_j$ where $C_j = \bigvee_{i} L_i$ is a disjunction of literals (variables or negated variables).
    The game $\cG_\phi$ is constructed as follows.
    For every variable $x$ we add two vertices (of player 1) $x$ and $\neg x$ connected by an exclusion edge.
    Moreover, for every clause $C_j$ we add a vertex, belonging to player 0, which has an outgoing edge into each literal $L_i$ occurring in it. There are no initial vertices, i.e.~$\init \ceq \emptyset$
    and the target set $T$ is the set of all clauses $C_j$ of $\phi$.
    Now, $S \in \target(\cG)$ if and only if $\bigwedge_{C_j \in S} C_j$ is satisfiable.
    In particular, $T \in\target(\cG_\phi)$ if and only if $\phi$ is satisfiable.
\end{proof}

The complexity is, of course, ascribed to the exclusion edges.

\begin{obs}
    \label{obs: inex without ex in P}
    The problem of deciding whether $X\in\target(\cG)$ for a finite inclusion-exclusion game $\cG$ without exclusion edges, i.e.~$\Eexcl = \emptyset$ is in \ptime and in fact complete for this class.
\end{obs}
\begin{proof}
    We describe the reduction from, and to, the well known $\ptime$ complete problem $\mathsf{Game}$.
    The input is a game graph $G = (V, V_0, V_1, E)$ and a vertex $v$ and the task is to determine whether player 0 has a winning strategy starting from $v$.
    Here again a player loses if he cannot make a move.

    It is easy to see that inclusion-exclusion games are an extension of this problem.
    For the converse direction add a fresh vertex $v$ of player 1 to the game that has outgoing edges to all $i\in I$ and $t\in X$.
    Further, turn all $t\in T\setminus X$ into player 0 vertices and remove all their outgoing edges.
    Now player 0 wins the $\mathsf{Game}$ problem on this graph if, and only if,  $X\in\target(\cG)$.
\end{proof}

\subsection{Model-Checking Games for Existential Second-Order Logic}

In this section we define model-checking games for formulae $\phi(X) \in \Sigma^1_1$ with a free relational variable.
These games are inclusion-exclusion games $\cG$ whose target sets $\target(\cG)$ are precisely the sets of relations $X$ that satisfy $\phi(X)$ under a fixed structure $\fA$.

\begin{defi}
    \label{def: inex game for Phi}
    Let $\fA$ be a $\tau$-structure and $\phi(X) =  \E \tR\phi'(X, \tR) \in \eso$ (in negation-normal form) where $\phi'(X, \tR) \in \fo(\tau \cup \{X, \tR\})$ using a free relation symbol $X$ of arity $r \ceq \arity{X}$.
    The game $\mcX{\fA}{\phi} \ceq (V, V_0, V_1, E, \init, T, \Eexcl)$ consists of the following components:
    \begin{itemize}
        \item $V \ceq \set{(\theta, s)}{\theta \in \subf(\phi'), s\colon \free(\theta) \ra A} \cup A^r$
        \item $\init \ceq \{(\varphi', \emptyset)\}$,
        \item$T \ceq A^r$,
        \item $\begin{aligned}[t]
        V_1 \,\ceq\, & \set{(\theta, s)}{\theta = \forall y \gamma \text{ or } \theta = \gamma_1 \land \gamma_2} \cup
        \set{(\gamma, s)}{\gamma \text{ is a $\tau$-literal and } \fA \models_s \gamma} \, \cup  \\
        & \set{(\gamma, s)}{\gamma = \neg X\tx \text{ or } \gamma \text{ is a $\{\tR \}$-literal}} \cup T\end{aligned}$,
        \item $V_0 \ceq V \setminus V_1$,
        \item $\begin{aligned}[t]
        E \,\ceq\, & \set{((\gamma \circ \theta,s),(\delta,s\res{\free(\delta)}))}{ \circ \in \{ \land, \lor \}, \delta \in  \{ \gamma, \theta \} } \, \cup \\
        & \set{((X\tx,s),s(\tx))}{X\tx \in \subf(\phi')} \, \cup  \\
        & \set{((Qx\gamma,s), (\gamma, s'))}{Q \in \{ \E, \A \}, s' = s[x \mapsto a], a \in A },
        \end{aligned}$
        \item $\Eexcl \ceq \set{((R_i \tx, s), (\neg R_i \ty, s'))}{s(\tx) = s'(\ty)} \cup \set{((\neg X \tx, s), \ta)}{s(\tx) = \ta}$.
    \end{itemize}
\end{defi}

\noindent
These games capture the behaviour of existential second-order formulae which provides us with the following theorem.

\begin{thm}
    \label{thm: eso <=> inex-game}
    $(\fA, X)\models\phi(X) \iff$ Player 0 has a winning strategy $\cS$ in $\mcX{\fA}{\phi}$ with $\target(\cS) = X$.
    Or, in other words: $\target(\mcX{\fA}{\phi}) = \set{Y \subseteq A^r}{(\fA, Y) \models \phi(X)}$.
\end{thm}
\begin{proof}
    % Direction ``=>'':
    First let $(\fA, Y)\models\phi = \E \tR \phi'(X,\tR)$.
    Then there exist relations $\tP$ such that $(\fA, Y, \tP) \models \phi'(X, \tR)$.
    So player $0$ wins the (first-order) model-checking game $\cG' \ceq \mc{(\fA, Y, \tP)}{\phi'(X, \tR)}$.
    Let $\cS' = (W',F')$ be a winning strategy for player 0 in $\cG'$ and $\cS \ceq (W,F)$ where $W \ceq W' \cup Y$ and
    $F \ceq F' \cup \set{((X\tx, s),\ta) \in V \times V}{\ta \in Y \text{ and } s(\tx) = \ta }$.
    Clearly we have that $\target(S) = W \cap T = Y$.
    To conclude this direction of the proof, we still need to prove that $\cS$ is indeed a winning strategy in $\cG\ceq\mcX{\fA}{\phi}$.
    The required properties for player $0$ and $1$, that is (\ref{def: winning strategy - player 0}) and (\ref{def: winning strategy - player 1}) of Definition \ref{def: winning strategy}, are inherited from $\cS'$ for every node of the form $(\theta, s) \in V\setminus T$ with $\theta \neq X \tx$.
    For nodes of the form $v = (X\tx, s) \in W$ we have that $v \in V_0$ and, because of $v \in W'$ and the fact that $\cS'$ is a winning strategy in $\cG'$, $s(\tx) \in Y$ must follow.
    As result, we have $s(\tx) \in \nh{\cS}{v}$.
    Since $T \subseteq V_1$ and are terminal positions, condition (\ref{def: winning strategy - player 1}) is trivially satisfied for all $v \in W\cap T$.

    Property (\ref{def: winning strategy - initial position}) is clearly satisfied, because $(\phi', \emptyset)$ is the initial position
    of $\cG'$ and, hence, $(\phi', \emptyset) \in W' \subseteq W$.
    In order to prove that the last remaining condition, the exclusion condition (\ref{def: winning strategy - exclusion edges}), is satisfied, consider any $(v,w) \in \Eexcl$.
    Then there are two possible cases:

    Case $(v,w) = ((R_i \tx, s), (\neg R_i \ty, s'))$ with $s(\tx) = s'(\ty)$: Then either $v$
    or $w$ is a losing position for player $0$ in $\cG'$. As a result, $W'$ does not contain both $v$ and $w$
    and, thus, neither does $W$.

    Case $(v,w) = ((\neg X \tx, s), \ta)$ and $s(\tx) = \ta$: If $v \in W$, then $v = (\neg X \tx, s) \in W'$ and,
    since $\cS'$ is a winning strategy for player $0$ in $\cG'$, it must be the case that $s(\tx) \notin Y$ which implies that
    $\ta = s(\tx) \notin W$.
%    If, on the other hand, $\ta \in W = W' \cup X$, then we have $\ta \in X$ (because $\ta \notin W'$ follows from
%    $W' \subseteq V(\cG') = V(\cG) \setminus A^{\arity{X}}$)
%    and $(\neg X\tx, s) \notin W'$ (otherwise $\cS'$ would not be a winning strategy for player 0 in $\cG'$) which, by construction of $\cS$, implies that $(\neg X\tx, s) \notin W$.
    So, $v \in W$ and $w \in W$ exclude each other.

    % Direction ``<='':
    For the converse direction, let $\cS = (W,F)$ now be a winning strategy for player $0$ in $\cG$ with $\target(\cS) = Y$.
    We have to show that $(\fA, Y) \models \exists \tR \phi'(X, \tR)$.
    Let $P_i \ceq \set{s(\tx)}{(R_i\tx, s) \in W}$.
    Furthermore, we define $\cS' \ceq \cS - A^{\arity{X}}$, that is the restriction of the strategy $\cS$ to $V(\cG') = V \setminus A^{\arity{X}}$.

    % Argumentation, dass S tatsächlich Gewinnstrategie ist:
    We prove that $\cS'$ is a winning strategy for player $0$ in the first-order model-checking game $\cG' \ceq \mc{(\fA, Y, \tP)}{\phi'(X, \tR)}$.
    First of all, the conditions for player $0$ and $1$ for non-terminal positions are inherited from $\cS$.
    For the same reason we also have $(\phi', \emptyset) \in V(\cS')$.
    We still need to prove that $\cS'$ contains only terminal positions that are winning for player $0$.
    This is inherited for all terminal positions that are not using any $R_i$ nor $X$.
    We will now investigate the other terminal positions, i.e.~positions of the form $((\neg)R_i\tx,s)$ or $((\neg)X\tx,s)$.
    Clearly, if $\cS'$ plays $(R_i \tx, s)$, then $(R_i \tx, s) \in W$ and $s(\tx) \in P_i$ (by definition of $P_i$)
    implying that $(\fA,Y,\tP) \models_s R_i\tx$ and, hence, $(R_i \tx, s)$ is a winning position for player $0$ in $\cG'$.
    In the case that $\cS'$ visits $(\neg R_i \tx, s)$, we know that $(\neg R_i \tx, s) \in W$ and, because $\cS$ respects the exclusion condition,
    a position of the form $(R_i \ty, s')$ with $s'(\ty) = s(\tx)$ cannot be in $W$.
    So, in this case, we have that $s(\tx) \notin P_i$ and, hence, $(\neg R_i \tx, s)$ is again a winning position for player $0$.
    If $\cS'$ contains $v \ceq (X\tx,s)$, then the edge $(v, s(\tx))$ is played by $\cS$ and, consequently, $s(\tx) \in W \cap T = \target(\cS) = Y$
    which shows that $v$ is a winning position for player $0$ in $\cG'$.
    If, however, $(\neg X\tx,s)$ is played by $\cS'$, then $(\neg X\tx,s) \in W$ and, due to exclusion condition, $s(\tx) \notin W$
    which proves that $s(\tx) \notin W \cap T = Y$ and, again, $(\neg X\tx, s)$ is a winning for player $0$ in $\cG'$.
    As a result, we have that $(\fA, Y) \models \phi$.
\end{proof}

\section{Characterising the Union Closed Formulae within Existential Second-Order Logic}
\label{sec:char U in eso}

In this section we investigate formulae $\phi(X)$ of existential \secorder logic that are closed under unions with respect to their free relational variable $X$.
Union closure, being a semantical property of formulae, is undecidable.
However, we present a \emph{syntactical} characterisation of all such formulae via the following normal form.

\begin{defi}
    \label{def: myopic formula}
    A formula $\phi(X)\in\eso$ is called \emph{myopic} if it has the shape $\phi(X) = \A\tx(X\tx\rightarrow \E \tR \phi'(X,\tR))$, where $\phi'\in\fo$ and $X$ occurs only positively\footnote{That is under an even number of negations.} in $\phi'$.
\end{defi}

\noindent
Variants of myopic formulae have already been considered for first-order logic \cite[Definition 19]{GalHel13} and for greatest fixed-point logics \cite[Theorem 24 and Theorem 26]{Graedel16},
but to our knowledge myopic $\eso$-formulae have not been studied so far.

Let $\union$ denote the set of all union closed $\eso$-formulae.
To establish the claim that myopic formulae are a normal form of $\union$ we need to show that all myopic formulae are indeed closed under unions and, more importantly, that every union closed formula can be translated into an equivalent myopic formula.
This translation is in particular constructive.

In order to strengthen the intuition on myopic formulae we start with a basic property.

\begin{prop}
    \label{prop: myopic => union closed}
    Every myopic formula is union closed.
\end{prop}

\begin{proof}
    Let $\phi = \A \tx (X\tx \ra \E \tR \phi'(X,\tR))$ and $(\fA, X_i)\models\phi$ for all $i\in I$.
    We claim that $(\fA, X)\models \phi$ for $X=\bigcup_{i\in I}X_i$.
    Let $\ta\in X_i\subseteq X$.
    By assumption $(\fA, X_i)\models_{\tx\mapsto\ta}\E\tR\phi'(X, \tR)$.
    A fortiori ($X$ occurs only positively in $\phi'$), we obtain $(\fA, X)\models_{\tx\mapsto\ta}\E\tR\phi'(X, \tR)$.
    Since $\ta$ was chosen arbitrarily, this property holds for all $\ta\in X$, hence the claim follows.
\end{proof}

Of course, the reverse direction is the difficult and interesting part.
The main theorem of this section combines both parts.

\begin{thm}
    \label{thm: myopic <=> union-closed}
    $\phi(X) \in \eso $ is union closed if and only if $\phi(X)$ is equivalent to some myopic $\eso$-formula.
\end{thm}

\noindent
We provide two proofs of the remaining direction, a direct one which is short and easy to understand, while it does not utilise the methods developed in Section \ref{sec: in/ex games}, which the second proof does.
This leads to a more difficult construction but provides more insight into the resulting formula.

The first variant shows that for every $\eso$-formula we can construct a myopic \quotes{companion} formula that basically is the union closed version of the original one.

\begin{thm}
    \label{thm: eso-sentence => myopic}
    For every $\phi(X)\in\eso$ there is myopic formula $\mu(X) \in \eso$ such that
    for every suitable structure $\fA$ and relation $X$ over $\fA$ holds
    \[
    (\fA,X) \models \mu(X) \Longleftrightarrow X \text{ can be written as } X = \bigcup_{i \in I} X_i \text{ where }
    \fA \models \phi(X_i) \text{ for every } i \in I.
    \]
\end{thm}

\begin{proof}
    Let $\mu(X) \ceq \A \tx (X\tx \ra \E Y (Y \subseteq X \land Y\tx \land \phi(Y)))$ where $Y\subseteq X$ is a shorthand for the formula $\A \ty (Y\ty \ra X \ty)$.
    Now, $\mu(X)$ is a myopic formula, since $X$ occurs only positively after the implication.
    We still need to prove the two directions of the claim.

    ``$\Longrightarrow$'': First assume that $(\fA, X) \models \mu(X)$.
    Then, for every $\ta \in X$, there exists some $Y_\ta \subseteq X$ with $(\fA, Y \mapsto Y_\ta) \models Y\ta \land \phi(Y)$.
    Thus, we have $\fA \models \phi(Y_\ta)$ and $\ta \in Y_\ta \subseteq X$ for every $\ta \in X$.
    The last property entails that $X = \bigcup_{\ta \in X} Y_\ta$.

    ``$\Longleftarrow$'':
    Now let $X = \bigcup_{i \in I} X_i$ where $\fA \models \phi(X_i)$ for every $i \in I$.
    For every $\ta \in X$, choose some index $i_\ta \in I$ with $\ta \in X_{i_\ta}$.
    Then we have $(\fA, X, Y \mapsto X_{i_\ta}) \models Y \subseteq X \land Y \ta \land \phi(Y)$ and, thus, $(\fA,X) \models \mu(X)$.
\end{proof}

\begin{cor}
    \label{cor: union closed => myopic}
    For every union closed formula $\phi(X)\in\eso$ there is an equivalent myopic formula $\mu(X) \in \eso$.
\end{cor}

\noindent
Theorem \ref{thm: myopic <=> union-closed} now follows from Proposition \ref{prop: myopic => union closed} and Theorem \ref{thm: eso-sentence => myopic}.
Towards greater insight into the myopic fragment we will now turn towards an alternative proof that every union closed formula is equivalent to a myopic one which utilises game theoretical methods.
For a fixed formula $\phi(X)$ the corresponding inclusion-exclusion game $\cG_X$ can be constructed by a first-order interpretation depending of course on the current structure.

\begin{lem}
    \label{lem: fo-interpretation for G_X}
    Let $\phi(X) = \exists \tR \phi'(X,\tR) \in \Sigma^1_1$ where $\phi' \in \fo(\tau \cup \{ X, \tR \})$ and $r \ceq \arity{X}$.
    Then there exists a quantifier-free interpretation $\cI$ such that $\cG_X(\fA, \phi) \cong \cI(\fA)$ for every structure $\fA$
    (with at least two elements).
\end{lem}
\begin{proof}
    The construction we use in this proof is similar to the one from \cite[Proposition 18]{Graedel16}.
    An equality type $e(\tv)$ over a tuple $\tv = (v_1, \dotsc, v_n)$ is a maximal consistent set of (in)equalities using only variables from $\tv$.
    Since equality types over finitely many variables are finite, we can, by slight abuse of notation, identify $e(\tv)$ with the formula $\bigwedge e(\tv)$.
    Let $n$ be chosen sufficiently large so that we can fix for every $\theta \in \subf(\phi') \cup \{T \}$ a unique equality type $e_\theta(\tv)$.

    Let $\tx = (x_1, \dotsc, x_m)$ be a tuple of variables such that for every subformula $\theta \in \subf(\phi')$ holds $\free(\theta) \subseteq \{ \tx \}$.
    For each variable $x_i \in \{\tx\}$ let $\iota(x_i) \ceq i$.
    A position $(\theta, s)$ of the game $\cG_X(\fA, \phi)$ will be encoded by an $(n+m)$-tuple of the form $(\tu, \ta)$ where $\tu$ has equality type $e_\theta$
    and $s(x_i) = a_i$ for every $x_i \in \free(\theta)$, while a position of the form $\ta \in T (= A^r)$ will be encoded by $(\tu, \ta \tb)$
    such that $\tu$ has equality type $e_T$ whereas $\tb \in A^{m-r}$ can be an arbitrary tuple.
    Now we are in the position to define the interpretation $\cI = (\delta, \epsilon, \psi_{V_0}, \psi_{V_1}, \psi_{E}, \psi_{I}, \psi_{T}, \psi_{\Eexcl})$:

    \begin{itemize}
        \item $\begin{aligned}[t] \delta(\tv, \ty)   \ceq \bigvee_{\theta \in \subf(\phi') \cup \{T\}} e_\theta(\tv) \end{aligned}$
        \item $\begin{aligned}[t]
                    \epsilon(\tv, \ty, \tw, \tz) \ceq \!\!\bigvee_{ \theta \in \subf(\phi')} (e_\theta(\tv) \land e_\theta(\tw) \land \bigwedge_{x_i \in \free(\theta)} y_i = z_i) \lor (e_T(\tv) \land e_T(\tw) \land  \bigwedge^r_{i = 1} y_i = z_i)
               \end{aligned}$
        \item $\psi_{V_1}(\tv, \ty) \ceq \bigvee_{(\theta, s) \in V_1(\cG_X(\fA, \phi))} e_\theta(\tv)  \lor e_T(\tv)$ and $\psi_{V_0}(\tv, \ty) \ceq \delta(\tv,\ty) \wedge \neg \psi_{V_1}(\tv, \ty)$.
        \item Let $R \ceq \set{(\theta, \theta')}{((\theta,s),(\theta',s')) \in E(\cG_X(\fA, \phi))}$. Then we define
        \begin{align*}
            \psi_{E}(\tv, \ty, \tw, \tz) \ceq &
            \bigvee_{(\theta, \theta') \in R} (e_\theta(\tv) \land e_{\theta'}(\tw) \land
            \bigwedge_{x_i \in \free(\theta) \cap \free(\theta')} y_i = z_i )  \\
            & \lor  \bigvee_{X\tu \in \subf(\phi')} (e_{X\tu}(\tv) \land e_T(\tw) \land \bigwedge^r_{i = 1} y_{\iota(u_i)} = z_i).
        \end{align*}
        \item Let $S \ceq \set{(R_i \tu, \neg R_i \tu')}{((R_i \tu,s),(\neg R_i \tu',s')) \in \Eexcl(\cG_X(\fA, \phi))}$\footnote{Since the direction of exclusion edges does not matter we assume here that they are all of this form.}. Then we define
        \begin{align*}
            \psi_{\Eexcl}(\tv, \ty, \tw, \tz) \ceq &
            \bigvee_{(R_i \tu, \neg R_i \tu') \in S} (e_{R_i \tu}(\tv) \land e_{\neg R_i \tu'}(\tw) \land
            \bigwedge^{\arity{R_i}}_{i = 1} y_{\iota(u_i)} = z_{\iota(u'_i)} )  \\
            & \lor  \bigvee_{\neg X\tu \in \subf(\phi')} (e_{\neg X\tu}(\tv) \land e_T(\tw) \land \bigwedge^r_{i = 1} y_{\iota(u_i)} = z_i).
        \end{align*}
        \item $\psi_{\init}(\tv, \ty) \ceq e_{\phi'}(\tv)$
        \item $\psi_{T}(\tv, \ty) \ceq e_T(\tv)$.
    \end{itemize}
    Now, for every $\fA$ (with at least two elements) we have that $\cI(\fA) \cong \cG_X(\fA, \phi)$.
\end{proof}

Remember that the construction used in the proof of Theorem \ref{thm: eso-sentence => myopic} in a way ``imports'' the original formula $\phi$ into $\mu$.
Therefore, the number of atoms using $X$ or some quantified second-order symbol in $\mu(X)$ is not bounded.
Using inclusion-exclusion games, we can construct $\mu(X)$ in a different way that limits the usage of these symbols.

\begin{thm}
    \label{thm: union closed eso => myopic stronger}
    For every union closed formula $\phi(X)\in\eso$ there is an equivalent myopic formula $\mu(X) \in \eso$
    where only $11$ literals use the symbol $X$ or some quantified second-order symbol.
\end{thm}

\begin{proof}
    Let $\phi(X) = \exists \tR \phi'(X,\tR) \in \eso(\tau)$ be closed under unions, $\fA$ be a $\tau$-structure
    and $\cG \ceq \mcX{\fA}{\phi}$ be the corresponding game.
    W.l.o.g.~$\fA$ has at least two elements.
    By Theorem \ref{thm: eso <=> inex-game}, we have that $\target(\cG) = \set{Y \subseteq A^r}{\fA \models \phi(Y)}$
    where $r \ceq \arity{X}$.
    Since $\phi(X)$ is union closed, it follows that $\target(\cG)$ is closed under unions as well.
    Now we observe that $\target(\cG)$ can be defined in the game $\cG$ by the following myopic formula:
    \begin{align*}
        \phitarget(X)   \ceq  & \A x (Xx \ra \psi_\target(X, x) ) \text{ where} \\
        \psi_\target(X,x) \ceq & \exists W (\phiwin(W) \land Wx \land \forall y (Wy \land Ty \rightarrow Xy) )
    \end{align*}
    Here $\phiwin$ is the first-order formula verifying winning strategies.
    Please note that $\phitarget$ is indeed a myopic formula, since $X$ occurs only positively in $\psi_\target$.
    Furthermore, there are $6$ $W$-atoms in $\phiwin$ and two additional $W$-atoms in $\psitarget$, while $X$
    occurs twice in $\phitarget$.
    In total, $X$ and $W$ are used exactly $10$ times.
    These $10$ atoms will also occur in the final formula $\mu$ that we are going to construct.
    This construction will use a $\fo$-interpretation, which will introduce an additional $W^\star$-atom in order to simulate the quantifier $\E W$ by a new quantifier $\E W^\star$.
    This is why, we will end up with exactly $11$ $\{ X, W^\star\}$-literals in $\mu$.

    \begin{clm}
        \label{claim: phitarget is correct}
        For every $X \subseteq A^r$, $(\cG, X) \models \phitarget(X) \iff X \in \target(\cG)$.
    \end{clm}

    \begin{proof}
        Assume that $(\cG, X) \models \phitarget(X)$.
        By construction of $\phitarget$, for every $\ta \in X$ there exists a winning strategy $\cS_\ta = (W_\ta, F_\ta)$ with
        $\ta \in W_\ta$ and $\target(\cS_{\ta}) = W_\ta \cap T \subseteq X$.
        It follows that $X = \bigcup_{\ta \in X} \target(\cS_{\ta})$.
        Since $\target(\cG)$ is closed under unions, we also obtain that $X \in \target(\cG)$.

        We want to remark that at this point the semantical property is translated into a syntactical one, as the formula only describes the correct winning strategy because the initial formula was closed under unions.

        To conclude the proof of Claim \ref{claim: phitarget is correct}, assume that $X \in \target(\cG)$.
        Then there exists a winning strategy $\cS = (W,F)$ for player $0$ with $\target(\cS) = X$.
        Thus, for the quantifier $\E W$ we can (for all $\ta\in X$) choose $\cS$, which, obviously, satisfies the formula.
    \end{proof}

    \noindent
    Recall the first-order interpretation $\cI$ (of arity $n+m$) from Lemma \ref{lem: fo-interpretation for G_X} with $\cI(\fA) \cong \cG$ for some coordinate map $h \colon \delta^\fA \ra V(\cG)$ and for every $\ta \in T(\cG)$,
    $h^{-1}(\ta) = \set{(\tu, \ta, \tb) \in A^{n+m}}{\fA \models e_T(\tu), \tb \in A^{m-r}}$
    where $e_T(x_1,\dots,x_n)$ is some quantifier-free first-order formula.
    By the interpretation lemma for $\eso$ (Lemma \ref{lem: interpretation lemma for sigma^1_1}), for every $X \subseteq T(\cG)$,
    \begin{equation} \label{eq: interpretation lemma result}
    (\fA, X^\star) \models \phitarget^\cI(X^\star) \Longleftrightarrow (\cG, X) \models \phitarget(X)
    \end{equation}
    where $X^\star \ceq h^{-1}(X)$ is a relation of arity $(n+m)$.
    Recall that every variable $x$ occurring in $\phitarget$ is replaced by a tuple $\tx$ of length $(n+m)$.
    Let $\tx = (\tu, \tv, \tw)$ where $|\tu| = n$, $|\tv| = r$ and $|\tw| = m-r$ and let
    \[\mu(X) \ceq \forall \tv (X\tv \ra \A \tu \A \tw (e_T(\tu) \ra \psi^\star(X, \tu, \tv, \tw)))\]
    where $\psi^\star$ is the formula that results from $\psi_\target^\cI$ by replacing every occurrence of $X^\star \tu' \tv' \tw'$ (where
    $|\tu'| = n$, $|\tv'| = r$ and $|\tw'| = m -r $) by the formula $e_T(\tu') \land X\tv'$.
    By construction, this is a myopic formula, because $X$ occurred only positively in $\psi^\cI$ and, hence, $X^\star$ (resp.~$X$)
    occurs only positively in $\psi^\cI_\target$ (resp.~$\psi^\star$).

    Recall that, in the game $\cG \cong \cI(\fA)$, every $X\subseteq T(\cG)$ is a unary relation over $\cG$, while the elements of $T(\cG)$ themselves are $r$-tuples of $A$.
    Furthermore, we have that $h^{-1}(X) \ceq \set{(\ta, \tb, \tc) \in A^n\times A^r \times A^{m-r}}{\fA \models e_T(\ta) \text{ and } \tb \in X}$.
    Because of this and $X^\star = h^{-1}(X)$, it follows that for every $s \colon \{ \tu', \tv', \tw' \} \to A$ holds
    \begin{equation} \label{eq: X vs Xstar}
    (\fA, X^\star) \models_s X^\star \tu' \tv' \tw' \iff \fA \models e_T(s(\tu')) \text{ and } s(\tv') \in X  \iff (\fA,X) \models_s e_T(\tu')\land X\tv'.
    \end{equation}
    By construction of $\psi^\star$, these are the only subformulae in which $\psi^\cI_\target$ and $\psi^\star$ differ from each other.
    As a result, the following claim is true:

    \begin{clm} \label{claim: construction psistar is correct}
        For every $X \subseteq A^r$ and every assignment $s \colon \free( \psi^\cI_\target) \ra A$, holds
        \[(\fA, X^\star) \models_s \psi^\cI_\target(X^\star, \tx) \iff (\fA, X) \models_s \psi^\star(X, \tx).\]
    \end{clm}
    Recall that $\tx = (\tu, \tv, \tw)$ where $|\tu| = n, |\tv| = r$ and $|\tw| = (m-r)$.
    Now we can see that
    \begin{align*}
        & (\fA, X^\star) \models \phitarget^\cI = \forall \tx (X^\star \tx \ra \psi^\cI_\target(X^\star, \tx))   \\
        \iff & (\fA, X^\star) \models_s \psi^\cI_\target(X^\star, \tx) \text{ for every } s \text{ with } s(\tx) \in X^\star  \\
        \iff & (\fA, X)^{\phantom{\star}} \models_s \psi^\star(X, \tx) \text{ for every } s \text{ with } s(\tx) \in X^\star  \;\; (\text{Claim \ref{claim: construction psistar is correct}})\\
        \iff & (\fA, X)^{\phantom{\star}}       \models_s \psi^\star(X, \tx) \text{ for every } s \text{ with } (\fA,X) \models_s e_T(\tu) \land X\tv \;\; (\text{due to (\ref{eq: X vs Xstar})})\\
        \iff  & (\fA, X)^{\phantom{\star}} \models \A \tu \A \tv \A \tw ((e_T(\tu) \land X\tv) \ra \psi^\star(X, \tu, \tv, \tw))) \equiv \mu.
    \end{align*}

    \noindent
    As a result, we have that $(\fA, X) \models \mu(X) \iff (\fA, X^\star) \models \phi^\cI_\target$. Putting everything together yields:
    \[                                                                                (\fA, X) \models \mu
    \Longleftrightarrow                                                               (\fA, X^\star) \models \phitarget^\cI
    \stackrel{\eqref{eq: interpretation lemma result}}{\Longleftrightarrow}           (\cG, X) \models \phitarget
    \stackrel{(\text{Claim \ref{claim: phitarget is correct}})}{\Longleftrightarrow} X \in \target(\cG)
    \stackrel{(\text{Theorem \ref{thm: eso <=> inex-game}})}{\Longleftrightarrow}     (\fA, X) \models \phi\]
    Thus, the constructed myopic formula $\mu(X)$ is indeed equivalent to $\phi(X)$.
\end{proof}

This construction can be applied to non union closed formulae as well, in which
case the statement becomes $(\fA, X)\models\mu$ if, and only if, $X = \bigcup_{i \in I}X_i$
such that $(\fA, X_i)\models\phi$ for all $i \in I$.
To see this replace Claim \ref{claim: phitarget is correct} by ``For every $X \subseteq A^r$,
$(\cG, X) \models \phitarget(X) \iff X = \bigcup_{i\in I}X_i$, where $X_i \in \target(\cG)$
for all $i\in I$''.
Or, in other words, $\mu$ is the ``smallest'' (w.r.t.~the set of models) union closed
formula which is implied by $\phi$: $\phi\models\mu$ and, more importantly, $\mu$
implies all union closed formulae $\nu$ which are implied by $\phi$, i.e.~if
$\phi\models\nu$ and $\nu$ is union closed then $\mu\models\nu$ holds.

\section{Union Games}
\label{sec:union games}

In the previous section we have characterised the union closed fragment of $\eso$ by means of a syntactic normal form.
Now we aim at a game theoretic description, which leads to the following restriction of inclusion-exclusion games that reveals \emph{how} union closed properties are assembled.

\begin{defi}
    \label{def: ugame}
    A \emph{union game} is an inclusion-exclusion game $\cG=(V, V_0, V_1, E, \init, T, \Eexcl)$ obeying the following restrictions.
    For every $t\in T$ the subgraph reachable from $t$ via the edges $E\setminus\Eincl$, that are the edges of $E$ that do \emph{not} go back into $T$, is denoted by $\gcomp{t}$.\footnote{Recall that $\Eincl \ceq E \cap (V \times T)$.}
    These components must be disjoint, that is $V(\gcomp{t})\cap V(\gcomp{t'}) = \emptyset$ for all $t\neq t' \in T$.
    %Denote by $\gcomp{U}$ for $U\subseteq T$ the graph $\gcomp{U}\ceq\bigcup_{t\in U}\gcomp{t}$.
    Furthermore, exclusion edges are only allowed between vertices of the same component, that is $\Eexcl \subseteq \bigcup_{t\in T} V(\gcomp{t})\times V(\gcomp{t})$.
    The set of initial positions is empty, i.e.~$\init = \emptyset$.
\end{defi}

\begin{figure}
    \begin{center}
        \begin{tikzpicture}
        \node (i1) at (0,2) {$t_1$};
        \node (i2) at (4.5,2) {$t_2$};
        \node (ik) at (10,2) {$t_k$};
        \node at (7.25,1) {$\dotsb$};
        \node (g1) at (0, -.5) {$\gcomp{t_1}$};
        \node (g2) at (4.5, -.5) {$\gcomp{t_2}$};
        \node (gk) at (10, -.5) {$\gcomp{t_k}$};

        \draw[rounded corners] (i1) -- (-1,0) -- (1,0) -- (i1);
        \draw[rounded corners] (i2) -- (3.5,0) -- (5.5,0) -- (i2);
        \draw[rounded corners] (ik) -- (9,0) -- (11,0) -- (ik);

        \draw (0,1) edge[->,out=0, in=180] (i2);
        \draw (4,0.5) edge[->,out=180, in=0] (i1);
        \draw (5,0) edge[->,out=180, in=340] (i1);
        \draw (10.25,1) edge[->,out=0, in=0] (ik);

        \node (a) at (-0.7,0.25) {$v$};
        \node (b) at (0.45,0.25) {$w$};
        \draw (a) edge[dashed,<->] (b);
        \node (x) at (4.5,1.25) {$x$};
        \node (y) at (5.125,.25) {$y$};
        \draw (x) edge[dashed,<->] (y);
        \node (u) at (9.5,.25) {$u$};
        \node (z) at (10.5,0.5) {$z$};
        \draw (u) edge[dashed,<->] (z);
        \end{tikzpicture}
    \end{center}
    \caption{A drawing of a union game.
        The target positions $T = \{t_1,\dotsc,t_k\}$ are at the top of the components $\gcomp{t}$, which are depicted by triangles.
        Recall that the inclusion edges, that are the edges going into target vertices, do not account for the reachability of the components $\gcomp{t}$.
        The exclusion edges $\Eexcl$ are drawn as dashed arrows and, as seen here, are allowed only inside a component.
    }
    \label{fig: ugame}
\end{figure}
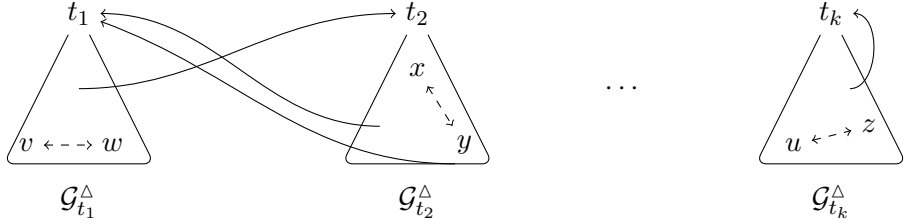

\noindent
See Figure \ref{fig: ugame} for a graphical representation of a union game.
Since the exclusion edges are only inside a component we can in a way combine different strategies into one, which is the reason the target set of a union game is closed under unions.

\begin{thm}
    \label{thm: ugames are union closed}
    Let $\cG$ be a union game and $(\cS_i)_{i\in J}$ be a family of winning strategies for player 0.
    Then there is a winning strategy $\cS$ for player 0 such that $\target(\cS) = \bigcup_{i\in J}\target(\cS_i)$.
    In other words, the set $\target(\cG)$ is closed under unions.
\end{thm}

\begin{proof}
    Let $\cS_i = (W_i, F_i)$ for $i\in J$.
    Let $U \ceq \bigcup_{i\in J}\target(\cS_i)$ and $f\colon U\to J$ be a function such that $t\in\target(\cS_{f(t)})$ for all $t\in U$.
    Define $\cS \ceq \bigcup_{t\in U}(\cS_{f(t)}\res{V(\gcomp{t})} + (E(\cS_{f(t)})\cap (V(\gcomp{t}) \times T)))$.
    In words, $\cS$ is defined on every component $\gcomp{t}$  with $t \in U$ as an arbitrary strategy $\cS_t$ that is defined on $\gcomp{t}$, including the inclusion edges leaving this component.
    By definition $\target(\cS) = U$ and, furthermore, $\cS$ is indeed a winning strategy since it behaves on every component $\gcomp{t}$ like $\cS_{f(t)}$ and there are no exclusion edges between different components.
\end{proof}

\begin{defi}
    \label{def: ugame for myopic Phi}
    Let $\mu(X)=\A \tx(X\tx\rightarrow\E \tR\phi(X, \tR, \tx))$ be a myopic $\tau$-formula where $\phi$ is in negation-normal form and $\fA$ be a $\tau$-structure.
    The union game $\mc{\fA}{\mu} \ceq (V, V_0, V_1, E, \init = \emptyset, T = A^{\arity{X}}, \Eexcl)$ is defined similarly to Definition \ref{def: inex game for Phi} with the difference being that for each $\ta\in A^{\arity{\tx}}$ we have to play on a copy of the game, so positions are now of the form $(\theta, s, \ta)$ instead of $(\theta, s)$, where $\theta\in\subf(\phi)$.
    The target vertices are the roots of these components, which is reflected by edges from $\ta$ to $(\phi, \tx\mapsto\ta, \ta)$.
    Because of this construction exclusion edges can only occur inside a component.
\end{defi}

\noindent
Notice that there are still edges from $(X\tx, s, \ta)$ to $s(\tx)$ --- the inclusion edges.
It is also worth mentioning that the empty set is always included in $\target(\mc{\fA}{\mu})$ for all myopic $\mu$ because $(\emptyset, \emptyset)$ is a
(trivial) winning strategy for player $0$.
This mimics the behaviour that in case $X=\emptyset$, the formula $\A\tx(X\tx\rightarrow\psi)$ is satisfied regardless of everything else.
The analogue of Theorem \ref{thm: eso <=> inex-game} holds for union games and myopic formulae.

\begin{prop}
    \label{prop: myopic (A, X) |= Phi <=> X in I(G^A_Phi)}
    Let $\fA$, $\mu$ and $\mc{\fA}{\mu}$ be as in Definition \ref{def: ugame for myopic Phi}.
    Then $(\fA, X)\models\mu \iff X\in\target(\mc{\fA}{\mu})$.
\end{prop}
\begin{proof}
    Assume $(\fA, X)\models\mu = \A \tx (X\tx \ra \E \tR\phi(X, \tR, \tx))$.
    Thus, for every $\ta\in X$ there exist relations $\tR_\ta$ such that $\fA\models\phi(X,\tR_\ta,\ta)$.
    Notice that every component $\gcomp{\ta}$ restricted to $V,V_0,V_1,E$ is essentially isomorphic to the first-order model-checking game $\cG((\fA,X,\tR_\ta), \phi)$
    (besides the additional node $\ta$, the only difference is that in the first-order model-checking game the vertices of the form $(X\ty, s, \ta)$ are terminal nodes where player 0 loses if and only if $s(\ty) \notin X$).
    Let $\cS^{\fo}_\ta$ be a winning strategy for player 0 in this game.
    Since either $\tb\in R_\ta$ or $\tb\notin R_\ta$ for all $\tb$ and $R$, one vertex $(R\tx,s,\ta)$ or $(\neg R\ty,s',\ta)$ with $s(\tx) = s'(\ty)$ is not visited by $\cS^{\fo}_\ta$.
    Let $\cS' \ceq \bigcup_{\ta\in X}\cS^{\fo}_\ta$ and $\cS \ceq \cS' + \Eincl\cap(V(\cS')\times V(\cS'))$.
    In words, the strategy $\cS$ combines all first-order strategies together and adds the reached inclusion edges.
    By definition $\target(\cS) = X$.
    Whenever a node of the form $(X\ty, s, \ta)$ is visited in $\cS$ we have that $s(\ty) \in X$ (because otherwise $\cS^{\fo}_\ta$ would not
    be a winning strategy for player 0) and hence $((X\ty, s, \ta),s(\ty)) \in \Eincl\cap(V(\cS')\times V(\cS'))$ is a move
    that is available to player $0$.
    That $\cS'$ satisfies the conditions for a winning strategy on the other nodes is inherited from the fact that the individual strategies are winning strategies on the first-order part.
    As pointed out before, each strategy does not visit an exclusion edge.

    For the contrary, let $\cS$ be a winning strategy with $\target(\cS)=X$.
    For every $\ta\in X$ let $\tb\in R_\ta$ if and only if there is some $(R\tx, s, \ta) \in V(\cS)$ with $s(\tx) = \tb$.
    We have to show that $\fA\models\phi(X,\tR_\ta,\ta)$ for all $\ta\in X$.
    But there is nothing to do here because $\cS\res{\gcomp{\ta}}$ induces a winning strategy for the first-order model-checking game for $\left<(\fA,X,\tR_\ta), \tx\mapsto\ta, \phi\right>$.
\end{proof}

\section{Myopic Fragment of Inclusion-Exclusion Logic}
\label{sec:myopic inex}

We turn our attention towards logics with team semantics in this section.
Our findings will be similar in nature to what we observed for existential second-order logic.
In fact, like the normal form of union closed $\eso$-formulae (see Section \ref{sec:char U in eso}) we present syntactic restrictions of inclusion-exclusion logic $\inexlogic$ that correspond precisely to the union closed fragment $\union$\footnote{We have defined $\union$ to be the set of all union closed $\eso$-formulae, by slight abuse of notation we use the same symbol here to denote the set of all $\inexlogic$-formulae that are closed under unions.}.
Analogously to myopic $\eso$-formulae we will also present a normal form for all union closed $\inexlogic$-formulae.

\begin{defi}
    \label{def: myopic in/ex-formula}
    A formula $\phi(\tx) \in \inexlogic$ is $\tx$-myopic, if the following conditions are satisfied:
    \begin{enumerate}
        \item The variables from $\tx$ are never quantified in $\phi$.
        \label{def: myopic in/ex-formula - quantifiers}
        \item Every exclusion atom occurring in $\phi$ is of the form $\tx \ty \excl \tx \tz$.
        \label{def: myopic in/ex-formula - exclusion atoms}
        \item Every inclusion atom occurring in $\phi$ is of the form $\tx \ty \incl \tx \tz$ or $\ty \incl \tx$, where the latter is \emph{only} allowed if it is not in the scope of a disjunction.
        \label{def: myopic in/ex-formula - inclusion atoms}
    \end{enumerate}
    Please note that $\phi(\tx)$ must not have any additional free variables besides $\tx$.
    We call atoms of the form $\tx\ty\incl\tx\tz$ or $\tx\ty\excl\tx\tz$ ($\tx$-)\emph{guarded} and $\ty\incl\tz$, respectively $\ty\excl\tz$, the corresponding \emph{unguarded} versions.
    Analogously, we call a formula $\psi$ the unguarded version of $\phi$, if $\psi$ emerges from $\phi$ by replacing every dependency atom by the respective unguarded version.
\end{defi}

\noindent
The intuition behind this definition is that every $\tx$-myopic formula can be evaluated componentwise on every team $X\res{\tx=\ta} = \set{s \in X}{s(\tx) = \ta}$ for all $\ta\in X(\tx)$.
For a formula $\phi$ let $T_\phi$ denote its syntax tree\footnote{Since we consider a tree instead of a DAG, identical subformulae may occur at different nodes.}.
A (team-)labelling of $T_\phi$ is a function $\lambda$ mapping every node $v_\psi$ to a team $\lambda(v_\psi)$ whose domain includes $\free(\psi)$.
In the following we write $\lambda(\psi)$ instead of $\lambda(v_\psi)$ if it is clear from the context which \emph{occurrence} of the subformula $\psi$ of $\phi$ is meant.
We call $\lambda$ a \emph{witness} for $\fA\models_X\phi$, if $\lambda(\phi)=X$ and the semantical rules of Definition \ref{def: fo ts} are satisfied (e.g.~$\lambda(\psi\lor\theta) = \lambda(\psi) \cup \lambda(\theta)$) and for every literal $\beta$ of $\phi$ we have $\fA\models_{\lambda(\beta)}\beta$.
By induction, if $\lambda$ is a witness for $\fA\models_X\phi$, then for every $\psi\in\subf(\phi)$ we have $\fA\models_{\lambda(\psi)}\psi$ and, moreover, $\fA\models_X\phi$ if and only if there is a witness $\lambda$ for $\fA\models_X\phi$.

\begin{prop}
    \label{prop: res inex components}
    Let $X$ be team over $\fA$ with $\dom(X) \supseteq \{ \tx, \tv, \tw \}$ and $\phi(\tx)$ be $\tx$-myopic.
    \begin{enumerate}
        \item $\fA\models_X \tx\tv\incl\tx\tw \iff \fA\models_{X\res{\tx=\ta}}\tv\incl\tw$ for all $\ta\in X(\tx)$ \label{prop: res inex components - guarded inclusion atoms}
        \item $\fA\models_X \tx\tv\excl\tx\tw \iff \fA\models_{X\res{\tx=\ta}}\tv\excl\tw$ for all $\ta\in X(\tx)$ \label{prop: res inex components - guarded exclusion atoms}
        \item For every subformula $\tv\incl\tx$ of $\phi$ and witness $\lambda$ for $\fA\models_X\phi$ we have $(\lambda(\tv\incl\tx))(\tx) = X(\tx)$. \label{prop: res inex components - other inclusion atoms}
    \end{enumerate}
\end{prop}
\begin{proof}
    We prove the first item.
    Let $\fA\models_X \tx\tv\incl\tx\tw$.
    That means for every assignment $s\in X$ there is another one, $s'\in X$, with $s(\tx\tv) = s'(\tx\tw)$.
    Thus $s(\tx)=s'(\tx)$ and therefore $s\in X\res{\tx=\ta} \iff s'\in X\res{\tx=\ta}$ from which $\fA\models_{X\res{\tx=\ta}}\tv\incl\tw$ follows for all $\ta\in X(\tx)$.

    Now assume $\fA\models_{X\res{\tx=\ta}}\tv\incl\tw$ for all $\ta\in X(\tx)$ and let $s\in X$ be an arbitrary assignment.
    Since $\fA\models_{X\res{\tx=s(\tx)}}\tv\incl\tw$ there is an assignment $s'\in X\res{\tx=s(\tx)}$ with $s'(\tw)=s(\tv)$.
    This means $s(\tx\tv)=s'(\tx\tw)$, and because $s$ was arbitrary, $\fA\models_X\tx\tv\incl\tx\tw$ follows.

    We prove the second item.
    $\fA \nmodels_{X \res{\tx = \ta}} \tv \excl \tw$ holds
    for some $\ta \in X(\tx)$ if and only if there are some $s,s' \in X$ with $\ta = s(\tx) = s'(\tx)$
    and $s(\tv) = s'(\tw)$, i.e.~$s(\tx\tv) = s'(\tx \tw)$ and hence $\fA \nmodels_X \tx \tv \excl \tx \tw$.
    Conversely, if $s(\tx\tv) = s'(\tx \tw)$ for some $s,s' \in X$, then $\fA \nmodels_{X\res{\tx = s(\tx)}} \tv \excl \tw$.

    The third item follows from the simple fact that $\tx$ is never quantified and that those atoms are not in the scope of a disjunction, hence the values of $\tx$ are preserved.
\end{proof}

The $\tx$-guarded version of a formula $\phi(\ty) \in \inexlogic$ in which the variables $\tx$ do not occur is the formula $\phi^\star(\tx, \ty)$ which results from $\phi(\ty)$
by replacing every inclusion/exclusion atom by its $\tx$-guarded variant, i.e.~$\tu \incl \tv$ would become $\tx \tu \incl \tx \tv$
and $\tu \excl \tv$ would be turned into $\tx \tu \excl \tx \tv$.
By property (\ref{prop: res inex components - guarded inclusion atoms}) and (\ref{prop: res inex components - guarded exclusion atoms}) of Proposition \ref{prop: res inex components}, the connection between a formula and its $\tx$-guarded version is the following.

\begin{lem} \label{lem: x-components of a team}
    Let $\phi^\star(\tx,\ty)$ be the $\tx$-guarded version of $\phi(\ty) \in \inexlogic$.
    Then $\fA \models_X \phi^\star(\tx, \ty) \Longleftrightarrow \fA \models_{X\res{\tx = \ta}} \phi(\ty)$ for every $\ta \in X(\tx)$.
\end{lem}

\noindent
Like union games an $\tx$-myopic formula is evaluated componentwise, which leads to the union closure of this fragment.

\begin{thm}
    \label{thm: myopic formulae are union closed}
    Let $\phi(\tx) \in \inexlogic$ be $\tx$-myopic and $\fA\models_{X_i}\phi$ for all $i\in I$.
    Then $\fA\models_X\phi$ for $X=\bigcup_{i\in I}X_i$.
\end{thm}
\begin{proof}
    Let $\lambda_i$ be a witness for $\fA \models_{X_i} \phi$ and every $i \in I$.
    For every $\ta\in X(\tx)$ choose $i_\ta \in I$ such that $\ta\in X_{i_\ta}(\tx)$.
    Define $\lambda(\psi) \ceq \bigcup_{\ta\in X(\tx)}\lambda_{i_\ta}(\psi)\res{\tx=\ta}$ for every $\psi \in \subf(\phi)$.
    We show that $\lambda$ is a witness for $\fA \models_X \phi$.
    It is not difficult to see that the requirements on witnesses for composite formulae are satisfied.
    We prove that the requirements for the literals are fulfilled as well.
    By the flatness property, first-order literals are satisfied by $\lambda$.

    We prove now that $\fA \models_{\lambda(\gamma)} \gamma$ for $\gamma = \tx\tv\incl\tx\tw$ or $\gamma = \tx\tv\excl \tx\tw$.
    Let $\gamma'$ be the corresponding \emph{unguarded} formula, that is the formula resulting from $\gamma$ by removing $\tx$, i.e.~we have
    $\gamma' = \tv \incl \tw$ or $\gamma' = \tv \excl \tw$.
    Due to Proposition \ref{prop: res inex components}, it suffices to prove that $\fA \models_{\lambda(\gamma)\res{\tx = \ta}} \gamma'$ is true
    for every $\ta \in (\lambda(\gamma))(\tx)$.
    Notice that $\lambda(\gamma)\res{\tx = \ta} = \lambda_{i_\ta}(\gamma) \res{\tx = \ta}$.
    Since $\lambda_{i_\ta}$ is a witness for $\fA \models_{X_{i_\ta}} \phi$, it must be the case that $\fA \models_{\lambda_{i_\ta}(\gamma)} \gamma$.
    By Proposition \ref{prop: res inex components}, it follows that
    $\fA \models_{\lambda_{i_\ta}(\gamma) \res{\tx = \tb}} \gamma'$ for every $\tb \in (\lambda_{i_\ta }(\gamma))(\tx)$.
    If $\ta \in (\lambda_{i_\ta}(\gamma))(\tx)$, then this implies that $\fA \models_{\lambda_{i_\ta}(\gamma) \res{\tx = \ta}} \gamma'$.
    Otherwise we have that $\lambda_{i_\ta}(\gamma) \res{\tx = \ta} = \emptyset$ and then $\fA \models_{\lambda_{i_\ta}(\gamma) \res{\tx = \ta}} \gamma'$ follows from the empty team property of $\inexlogic$.
    In both cases, $\fA \models_{\lambda_{i_\ta}(\gamma) \res{\tx = \ta}} \gamma'$ holds as desired, which concludes the proof of $\fA \models_{\lambda(\gamma)} \gamma$.

    We still need to prove that $\fA \models_{\lambda(\gamma)} \gamma$ for literals of the form $\gamma = \tv \incl \tx \in \subf(\phi)$.
    Towards this end, let $s \in \lambda(\tv \incl \tx)$ and $\tb \ceq s(\tv)$.
    By definition of $\lambda$, there is some $\ta \in X(\tx)$
    such that $s \in \lambda_{i_\ta}(\tv \incl \tx) \res{\tx = \ta}$.
    Since $\fA \models_{\lambda_{i_\ta}(\tv \incl \tx)} \tv \incl \tx$ and $\tb = s(\tv) \in \lambda_{i_\ta}(\tv \incl \tx)(\tv)$, it follows that $\tb \in (\lambda_{i_\ta}(\tv \incl \tx))(\tx)$.
    By Proposition \ref{prop: res inex components} we have $(\lambda_{i_\ta}(\tv \incl \tx))(\tx) = X_{i_\ta}(\tx)$, wherefore $\tb \in X_{i_\ta}(\tx) \subseteq X(\tx)$ and, consequently, we
    have chosen some index $i_\tb \in I$ with $\tb \in X_{i_\tb}(\tx)$.
    By Proposition \ref{prop: res inex components} again, it follows that $X_{i_{\tb}}(\tx) = (\lambda_{i_\tb}(\tv \incl \tx))(\tx)$.
    So there is some $s' \in \lambda_{i_\tb}(\tv \incl \tx)$ with $s'(\tx) = \tb = s(\tv)$
    and thus $s' \in \lambda_{i_\tb}(\tv \incl \tx)\res{\tx = \tb} \subseteq \lambda(\tv \incl \tx)$. This concludes the proof of $\fA \models_{\lambda(\tv \incl \tx)} \tv \incl \tx$.
\end{proof}

\noindent
It remains to prove that indeed every union closed formula $\phi$ of $\inexlogic$ is equivalent to some $\tx$-myopic formula.

\begin{thm}
    \label{thm: myopic inex - more general}
    Let $\phi(\tx) \in \inexlogic$.
    There is an $\tx$-myopic formula $\mu(\tx)$ such that for all suitable structures $\fA$ and teams $X$ with $\dom(X) = \{ \tx \}$ holds
    \[ \fA \models_X \mu(\tx) \Longleftrightarrow X \text{ can be written as } X = \bigcup_{i\in I} X_i \text{ where } \fA \models_{X_i} \phi(\tx) \text{ for every } i \in I. \]
\end{thm}

\begin{proof}
    Let $\ty$ be a fresh tuple of variables.
    Let $\phi^\star(\tx, \ty)$ be the $\tx$-guarded version of $\phi(\ty)$, i.e.~$\phi^\star$ results from $\phi$ by first replacing the variables $\tx$ by the new variables $\ty$ and then by adding $\tx$ on both sides of every inclusion or exclusion atoms occurring in $\phi(\ty)$. We define
    \[ \mu(\tx) \ceq \E \ty (\ty \incl \tx \land \tx \tx \incl \tx \ty \land \phi^\star(\tx, \ty)), \]
    which is $\tx$-myopic.
    We still need to prove the two directions of the claim.

    ``$\Longrightarrow$'': First assume $\fA \models_X \mu(\tx)$.
    Then there exists a team $Y$ of the form $Y = X[\ty \mapsto F]$ for some function $F\colon X \to \potne{A^{|\tx|}}$ such that
    \[\fA \models_Y \ty \incl \tx \land \tx \tx \incl \tx \ty \land \phi^\star(\tx, \ty).\]
    Thus we have $Y(\ty) \subseteq Y(\tx) = X(\tx)$ and, due to Lemma \ref{lem: x-components of a team}, $\fA \models_{Y\res{\tx = \ta}} \tx \incl \ty \land \phi(\ty)$ for every $\ta \in Y(\tx)$, because $\tx\tx \incl \tx\ty \land \phi^\star(\tx,\ty)$ is the $\tx$-guarded version of $\tx \incl \ty \land \phi(\ty)$.
    We can deduce, for every $\ta \in Y(\tx)$, that $\{ \ta \} = Y\res{\tx = \ta} (\tx) \subseteq Y\res{\tx = \ta}(\ty) \subseteq Y(\ty) \subseteq Y(\tx)$.  This implies that
    \begin{equation} \label{eq: a in Ya subseteq Y(x)}
        \ta \in Y\res{\tx = \ta}(\ty) \subseteq Y(\tx) \text{ for every } \ta \in Y(\tx).
    \end{equation}
    For every $\ta \in Y(\tx)$, let $X_\ta$ be the team with $\dom(X_\ta) = \{ \tx \}$ and $X_\ta(\tx) = Y\res{\tx = \ta}(\ty)$.
    Because of this construction and due to $\fA \models_{Y\res{\tx = \ta}} \phi(\ty)$, it follows that $\fA \models_{X_\ta} \phi(\tx)$.
    Furthermore, \eqref{eq: a in Ya subseteq Y(x)} and $Y(\tx) = X(\tx)$ lead to
    \[
    \ta \in X_\ta(\tx) \subseteq X(\tx) \text{ for every } \ta \in X(\tx)
    \]
    and, consequently, to $\bigcup_{\ta \in X(\tx)} X_\ta(\tx) = X(\tx)$.
    Because of $\dom(X) = \{ \tx \} = \dom(X_\ta)$, we indeed have $X = \bigcup_{\ta \in X(\tx)} X_\ta$ where $\fA \models_{X_\ta} \phi(\tx)$ for every $\ta \in X(\tx)$.

    ``$\Longleftarrow$'': For the converse direction, we assume that $X$ can be written as $X = \bigcup_{i \in I} X_i$ where $\fA \models_{X_i} \phi(\tx)$ for every $i \in I$.
    Our goal is to prove that $\fA \models_X \mu(\tx)$.
    Towards this end, for every $\ta \in X(\tx)$, we choose some index $i_\ta \in I$ with $\ta \in X_{i_\ta}(\tx)$.
    Let $F\colon X \to \potne{A^{|\ty|}}$ be defined by $F(s) \ceq X_{i_{s(\tx)}}(\tx)$.
    Then we have
    \begin{equation} \label{eq: extention function F}
        s(\tx) \in F(s) \subseteq X(\tx) \text{ for every } s \in X.
    \end{equation}
    Let $Y \ceq X[\ty \mapsto F]$.
    By construction, it follows that $F(s) = Y\res{\tx = s(\tx)}(\ty)$ for every $s \in X$.
    Furthermore, we also have
    \[Y(\ty) = \bigcup_{s \in X} F(s) \underset{\eqref{eq: extention function F}}{=} X(\tx) = Y(\tx).\]
    In particular, this implies $Y(\ty) \subseteq Y(\tx)$ and, hence, $\fA \models_Y \ty \subseteq \tx$.
    Now, in order to prove $\fA \models_Y \tx \tx \incl \tx \ty \land \phi^\star(\tx, \ty)$ we will use Lemma \ref{lem: x-components of a team}. This means that we only need to prove that $\fA \models_{Y\res{\tx = \ta}} \tx \incl \ty \land \phi(\ty)$ for every $\ta \in Y(\tx)$.
    Towards this end, pick any $\ta \in Y(\tx)$.
    Because of $Y(\tx) = X(\tx)$, there must be an assignment $s_\ta \in X$ with $s_\ta(\tx) = \ta$.
    We clearly have
    \[Y\res{\tx = \ta}(\tx) = \{ \ta \} = \{ s_\ta(\tx) \} \underset{\eqref{eq: extention function F}}{\subseteq} F(s_\ta) = Y\res{\tx = s_\ta(\tx)}(\ty)  = Y\res{\tx = \ta}(\ty) \]
    and, thus, $\fA \models_{Y\res{\tx = \ta}} \tx \incl \ty$.
    By assumption, we also know that $\fA \models_{X_{i_\ta}} \phi(\tx)$.
    Because of $Y\res{\tx = \ta}(\ty) = F(s_\ta) = X_{i_{s_\ta(\tx)}}(\tx) = X_{i_\ta}(\tx)$, this implies that $\fA \models_{Y\res{\tx = \ta}} \phi(\ty)$.
    Therefore, we indeed have $\fA \models_{Y\res{\tx = \ta}} \tx \subseteq \ty \land \phi(\ty)$ for every $\ta \in Y(\tx)$.
    As a result, we have $\fA \models_Y \ty \incl \tx \land \tx \tx \incl \tx \ty \land \phi^\star(\tx, \ty)$ which leads to $\fA \models_X \mu(\tx)$.
\end{proof}

As we have already seen in Theorem \ref{thm: myopic <=> union-closed}, every union closed formula of existential \secorder logic is equivalent to some myopic $\eso$-formula.
Moreover, it is well known that every $\inexlogic$-formula can be translated into an equivalent $\eso$-formula \cite{Gal12}.
Such a formula can be expressed as an $\tx$-myopic one of the form $\E \ts (\ts \incl \tx \land \psi)$ where $\psi$ uses only $\tx$-guarded
atoms.

\begin{lem}\label{lem: res inex components - normalform}
    Let $\phi(\tx) \in \inexlogic$ be an $\tx$-myopic formula of the form $\E\ts(\ts\incl\tx \land \psi)$, where in $\psi$ no inclusion atoms of the form $\ty\incl\tx$ occur. Then $\fA\models_X\phi$ if and only if there exists $F\colon X\to \potne{A^{|\tx|}}$ such that $F(s)\subseteq X(\tx)$ for every $s\in X$ and $\fA\models_{X[\ts\mapsto F]\res{\tx=\ta}}\psi'$ for all $\ta\in X(\tx)$, where $\psi'$ is the unguarded version of $\psi$.
\end{lem}

\begin{proof}
    By induction on $\psi$ and applying Proposition \ref{prop: res inex components}.
\end{proof}

We present two different proofs for the next theorem, which bring a myopic $\eso$-formula into this normal form.
The following proof is based on methods of Galliani, Kontinen and \jouko \cite{Gal12, KonVaa09} while the other one resembles the proof of Theorem \ref{thm: union closed eso => myopic stronger} and can be found in Section \ref{sec: alternative proof}.

\begin{thm}
    \label{thm: myopic eso translates into myopic inex}
    Let $\phi(X)$ be a myopic $\eso$-formula.
    There is an equivalent $\tx$-myopic formula of $\inexlogic$ where $|\tx| = \arity{X}$.
\end{thm}

\begin{proof}
    First of all let us introduce a normal form of myopic $\eso$-formulae.
    Since in myopic formulae the variable $X$ may occur only positively in the subformula $\phi'$, we can transform every $\A\tx(X\tx\ra\E\tR\phi'(\tR,X,\tx))$ into the equivalent formula $\A\tx(X\tx\ra\E S(S\subseteq X\land\E\tR \phi'(\tR,S,\tx)))$, where $S\subseteq X$ is a shorthand for $\A \ty (S\ty\ra X\ty)$.
    We now apply the Skolem-normal form of $\eso$-formulae to $\E\tR\phi'(\tR,S,\tx)$, which yields the formula $\sigma(S, \tx) \ceq\E\tf\A\ty((f_1(\tw)=f_2(\tw)\lra S\tw) \land \psi(\tf,\tx,\ty))$, where $\psi$ is a quantifier-free first-order formula and $\tw$ is a subtuple of $\ty$ and, moreover, every $f_i$ occurs in $\sigma$ only with a unique tuple $\tw_i$ (consisting of pairwise different variables) as argument, that is $f_i(\tw_i)$ (see \cite{KonVaa09} where an analogous construction is made).
    The original formula can thus be transformed into $\A\tx(X\tx\ra\E S(S\subseteq X\land \sigma(S, \tx)))$.
    Similarly to \cite{Gal12} we embed $\sigma(S, \tx)$ into inclusion-exclusion logic as $\theta(\ts,\tx)\ceq\A\ty\E\tz\big(\bigwedge_i\dep(\tx\tw_i,z_i)\land((\tx\tw\incl\tx\ts\land z_1=z_2)\lor(\tx\tw\excl\tx\ts\land z_1\neq z_2))\land\psi'(\tx,\ty,\tz)\big)$.
    Here $\psi'$ is obtained from $\psi$ by simply replacing every occurrence of $f_i(\tw_i)=f_j(\tw_j)$ by $z_i=z_j$.
    The only difference in our case is that every dependency atom is $\tx$-guarded due to the fact that the subformula at hand is inside the scope of the universally quantified variables $\tx$ in $\A\tx(X\tx \ra \dots)$.
    Notice that dependence atoms of the form $\dep(\tx\tw_i, z_i)$ can also be regarded as $\tx$-myopic.
    Formally, we can embed such an atom into exclusion logic via the formula $\A v(\tx\tw_i v\excl \tx\tw_i z_i \lor z_i=v)$, which has the intended shape \cite{Gal12}.
    The whole formula $\phi(X)$ thus translates into $\mu(\tx)\ceq\E\ts(\ts\incl\tx \land \theta(\ts,\tx))$.
    Let $\theta'(\ts,\tx)$ be the unguarded version of $\theta(\ts,\tx)$.
    Analogously to the argumentation of Galliani \cite{Gal12} by additionally making use of Proposition \ref{prop: res inex components}, we see that $(\fA, Y\res{\tx=\ta}(\ts)) \models_{\tx\mapsto\ta} \sigma(S, \tx)$ if and only if $\fA\models_{Y\res{\tx=\ta}} \theta'(\tx)$ for $\ta \in Y(\tx)$, where $Y$ is a team with domain $\{\ts,\tx\}$ (here the variable $S$ takes the role of the team).
    Using Lemma \ref{lem: res inex components - normalform} we have $\fA\models_X\mu(\tx)$ if and only if there is a function $F\colon X\to\potne{A^{\arity{\ts}}}$ such that $F(s)\subseteq X(\tx)$ for every $s \in X$ and $\fA\models_{X[\ts\mapsto F]\res{\tx=\ta}} \theta'(\ts,\tx)$ for all $\ta\in X(\tx)$, which again holds if and only if there exists such an $F$ and $(\fA,F(s))\models_{t}\sigma(S, \tx)$ for all $t\in X$, but this just means $(\fA, X(\tx)) \models \A\tx(X\tx\ra \E S(S\subseteq X \land \sigma(S, \tx)))$.
\end{proof}

\begin{cor}[Normal form of myopic-$\inexlogic$]
    \label{cor: normalform myopic inex}
    Let $\phi(\tx)$ be a union closed formula of $\inexlogic$.
    There is a logically equivalent $\tx$-myopic formula $\psi(\tx) = \E\ts(\ts\incl\tx\land\theta)$ where in $\theta$ only $\tx$-guarded dependency atoms occur.
\end{cor}

\subsection{Alternative Proof}
\label{sec: alternative proof}

The proof of Theorem \ref{thm: myopic eso translates into myopic inex} uses a translation into a special Skolem-normal form
that in the end produces a lot of guarded in-/exclusion atoms.
We can prevent this by using a different proof technique that similar to the proof of Theorem \ref{thm: myopic eso translates into myopic inex}
exploits that winning strategies of union games are definable in the myopic fragment of $\inexlogic$-formulae
and that the model-checking games of myopic formulae are first-order interpretable.
The rest of this section is organised as follows.
First, we define target sets of a union game by a myopic $\inexlogic$-formula, then
we adapt the interpretation lemma for logics with team semantics and finally present the alternative proof of Theorem \ref{thm: myopic eso translates into myopic inex}.

\begin{exa}
    Let us demonstrate that (vertex sets of) winning strategies of \emph{general} inclusion-exclusion games can be defined in $\inexlogic$:
    \begin{align*}
        \psiwin(y)   &\ceq \psiinit(y) \land \psimove(y) \land \psiexcl(y) \text{ where} \\
        \psiinit(y)  &\ceq \A z (Iz \ra z \incl y) \\
        \psimove(y)  &\ceq \E z ([(V_0y  \land \E z'(Eyz' \land z' \incl z)) \lor (V_1 y \land \A z' (Eyz' \ra z' \incl z))] \land z \incl y) \\
        \psiexcl(y)  &\ceq \A z ((\Eexcl yz \lor \Eexcl zy) \ra y \excl z)
    \end{align*}
    It is not difficult to verify that these formulae are just expressing the conditions for winning strategies (cf.~Definition \ref{def: winning strategy}).
    More formally, we have the following claim:

    \begin{clm} \label{claim: psiwin}
        Let $\cG$ be an inclusion-exclusion game
        and $Y$ be a non-empty team over $\cG$ with $y \in \dom(Y)$.
        Then $\cG \models_Y \psiwin(y)$ if and only if $Y(y)$ is the vertex set of a winning strategy for player $0$ in $\cG$.
    \end{clm}
    \noindent
    With these formulae at hand, it is easy to define the target sets in $\inexlogic$:
    \[ \psitarget(z) \ceq Tz \land \E y (\psiwin(y) \land z \incl y \land (Ty \ra y \incl z)).\]

    \begin{clm} \label{claim: psitarget}
        Let $\cG$ be an inclusion-exclusion game and let $X$ be a non-empty team over $\cG$ with $z \in \dom(X)$.
        Then $\cG \models_X \psitarget(z)$ if and only if $X(z) \in \target(\cG)$.
    \end{clm}

    \begin{proof}
        ``$\Longrightarrow$'':
        Let $\cG \models_X \psitarget(z)$. Then $X(z) \subseteq T(\cG)$ and $\cG \models_Y \psiwin(y) \land z \incl y \land (Ty \ra y \incl z)$ where $Y \ceq X[y \mapsto F]$ for some $F \colon X \to \potne{V(\cG)}$.
        From $\cG \models_Y z \incl y \land (Ty \ra y \incl z)$ we obtain that
        $Y(z) \subseteq Y(y)$ and $Y(y) \cap T(\cG) = (Y\res{Ty})(y) \subseteq (Y\res{Ty})(z)$.
        By Claim \ref{claim: psiwin} and $\cG \models_Y \psiwin(y)$, we have $Y(y) = V(\cS)$ for some winning strategy $\cS$.
        We prove that $Y(z) = Y(y) \cap T(\cG)$. The direction ``$\supseteq$'' follows from
        $Y(y) \cap T(\cG) \subseteq (Y\res{Ty})(z) \subseteq Y(z)$,
        while the direction ``$\subseteq$'' is entailed by $Y(z) \subseteq Y(y)$ and $Y(z) = X(z) \subseteq T(\cG)$.
        Thus, $X(z) = Y(z) = Y(y) \cap T(\cG) = \target(\cS) \in \target(\cG)$.

        ``$\Longleftarrow$'': Now let $X(z) \in \target(\cG)$.
        Then there is some winning strategy $\cS$ with $\target(\cS) = X(z)$.
        So $\cG \models_X Tz$.
        By letting $Y \ceq X[y \mapsto V(\cS)]$, we also obtain $\cG \models_Y \psiwin(y)$, by Claim \ref{claim: psiwin} and $Y(y) = V(\cS)$,
        and $\cG \models_Y z \incl y$, since $Y(z) = X(z) = \target(\cS) \subseteq  V(\cS) = Y(y)$.
        We still need to prove that $\cG \models_Y Ty \ra y \incl z$.
        Towards this end, consider any $s \in Y\res{Ty}$.
        Then $s(y) \in Y(y) \cap T(\cG) = \target(\cS) = X(z)$ and, thus, there exists some $s' \in X$ with $s'(z) = s(y)$.
        Let $s'' \ceq s'[y \mapsto s(y)]$.
        It follows that $s'' \in {Y\res{Ty}}$, because we have $s''(y) = s(y) \in \target(\cS) = V(\cS) \cap T(\cG)$.
        So we have $s''(z) = s'(z) = s(y)$ and $s'' \in Y\res{Ty}$, which concludes the proof of $\cG \models_Y Ty \ra y \incl z$.
        All in all, this proves that $\cG \models_X \psitarget$.
    \end{proof}
\end{exa}

\noindent
In the last example, we have learned that target sets of general inclusion-exclusion games can be expressed
in full $\inexlogic$. In particular, this formula also defines the target sets of union games.
But is it possible to do the same in the \emph{myopic fragment} of $\inexlogic$?

Towards giving a positive answer to this question, let $\psitarget^G(x,z)$ be the corresponding $x$-guarded version of
$\psitarget(z)$.
Recall that this just means that we add $x$ on both sides of every occurring in-/exclusion atom.
For example, the inclusion atom $z\incl y$ occurring in the subformula $\psiinit$ will be transformed into $xz \incl xy$
when constructing $\psitarget^G$.

\begin{exa}
    Consider the $x$-myopic formula $\theta_\target(x)$ given by
    \[\theta_\target(x) \ceq \E x' (x' \incl x \land \zeta(x,x')) \text{ where }
    \zeta(x,x') \ceq  xx \incl xx'\land \psitarget^G(x,x').\]
    We claim (and prove) that $\theta_\target(x)$ defines the target sets in \emph{union} games which, more formally, means that
    for every team $X$ over some union game $\cG$ with $x \in \dom(X)$ holds $\cG \models_X \theta_\target \iff X(x) \in \target(\cG)$.\footnote{Since union games do not have any
        initial vertices, the formula $\theta_\target$ could be simplified by removing the subformula $\psiinit^G$ which is in a union game just trivially satisfied.}
    To see this, let $\zeta'$ be the corresponding unguarded version of $\zeta$, which turns out to be the following formula:
    \[ \zeta'(x,x') = x \incl x' \land \psitarget(x') \]

    \begin{clm} \label{claim: zeta unguarded}
        Let $\cG$ be an inclusion-exclusion game and let $X$ be a non-empty team over $\cG$ with $x,x' \in \dom(X)$.
        Then $\cG \models_X \zeta'(x,x')$ if and only if $X(x) \subseteq X(x') \in \target(\cG)$.
    \end{clm}

    \begin{proof}
        Follows immediately from Claim \ref{claim: psitarget}.
    \end{proof}
    \noindent
    With this claim at hand, we can prove that $\theta_\target$ really defines what we promised:

    \begin{clm} \label{claim: thetatarget is correct}
        Let $\cG = (V, V_0, V_1, E, \init, T, \Eexcl)$ be a \emph{union} game
        and let $X$ be a team with  $x \in \dom(X)$. Then $\cG \models_X \theta_\target \iff X(x) \in \target(\cG)$.
    \end{clm}

    \begin{proof}
        For $X = \emptyset$ the claim is true, because $\emptyset \in \target(\cG)$ is true
        for every union game and $\cG \models_\emptyset \theta_\target$ is due to the empty team property.
        Let $X \neq \emptyset$.
        Applying Lemma \ref{lem: res inex components - normalform} immediately yields the equivalence
        of the following statements:
        \begin{enumerate}
            \item $\cG \models_X \theta_\target = \E x' (x' \incl x \land \zeta(x,x'))$ \label{claim: thetatarget is correct - (a)}
            \item There is a function $F \colon X \to \potne{V}$ such that \label{claim: thetatarget is correct - (b)}
            \begin{enumerate}
                \item $F(s) \subseteq X(x)$ for every $s \in X$, and \label{claim: thetatarget is correct - (b)(i)}
                \item $\cG \models_{X_t} \zeta'(x,x')$ where $X_t \ceq X[x' \mapsto F]\res{x = t}$ for every $t \in X(x)$. \label{claim: thetatarget is correct - (b)(ii)}
            \end{enumerate}
        \end{enumerate}
        Next we will prove that the following propositions are also equivalent:
        \begin{enumerate} \setcounter{enumi}{2}
            \item For every $t \in X(x)$ exists a winning strategy $\cS_t$ with $t \in \target(\cS_t) \subseteq X(x)$. \label{claim: thetatarget is correct - (c)}
            \item $X(x) \in \target(\cG)$. \label{claim: thetatarget is correct - (d)}
        \end{enumerate}
        \begin{proof}
            If $X(x) \in \target(\cG)$, then there is a strategy $\cS$ with $\target(\cS) = V(\cS) \cap T = X(x)$ which
            in particular implies that $t \in \target(\cS) \subseteq X(x)$ for every $t \in X(x)$.

            For the converse direction assume (\ref{claim: thetatarget is correct - (c)}).
            Since $\cG$ is a union game, we are allowed to use Theorem \ref{thm: ugames are union closed}
            to combine the family $(\cS_t)_{t\in X(x)}$ into a single winning strategy $\cS$ with $\target(\cS) = \bigcup_{t \in X(x)} \target(\cS_t)$ which is, due to assumptions about $\target(\cS_t)$,
            equal to $X(x)$.
        \end{proof}

        So we have already established that $(\ref{claim: thetatarget is correct - (a)}) \iff (\ref{claim: thetatarget is correct - (b)})$ and $(\ref{claim: thetatarget is correct - (c)}) \iff (\ref{claim: thetatarget is correct - (d)})$, but our goal was to show that $(\ref{claim: thetatarget is correct - (a)}) \iff (\ref{claim: thetatarget is correct - (d)})$ which is exactly what Claim \ref{claim: thetatarget is correct} states.
        Thus, in order to complete our proof of Claim \ref{claim: thetatarget is correct}, we just have to verify the missing link $(\ref{claim: thetatarget is correct - (b)}) \iff (\ref{claim: thetatarget is correct - (c)})$.

        ``$(\ref{claim: thetatarget is correct - (b)}) \Longrightarrow (\ref{claim: thetatarget is correct - (c)})$'':
        Suppose that there is some function $F \colon X \to \potne{V}$ such that (\ref{claim: thetatarget is correct - (b)(i)}) and (\ref{claim: thetatarget is correct - (b)(ii)}) are true.
        So $\cG \models_{X_t} \zeta'(x,x')$ for every $t \in X(x)$ which, by Claim \ref{claim: zeta unguarded}, yields that
        $X_t(x) \subseteq X_t(x') \in \target(\cG)$.
        By definition of $X_t$ in (\ref{claim: thetatarget is correct - (b)(ii)}), we have $X_t(x) = (X[x' \mapsto F]\res{x = t})(x) = \{ t \}$
        and, consequently, $t \in X_t(x') \in \target(\cG)$ which, by definition of $\target(\cG)$,
        leads to the existence of winning strategies $\cS_t$ with $X_t(x') = \target(\cS_t)$ for every $t \in X(x)$.
        Because of (\ref{claim: thetatarget is correct - (b)(i)}) we can also conclude that $X_t(x') = \bigcup_{s \in X\res{x=t}} F(s) \subseteq X(x)$.
        As a result, we obtain $t \in \target(\cS_t) \subseteq X(x)$ for every $t \in X(x)$ as desired.

        ``$(\ref{claim: thetatarget is correct - (b)}) \Longleftarrow (\ref{claim: thetatarget is correct - (c)})$'': We assume now that for every $t \in X(x)$ there exists some winning strategy
        $\cS_t$ with $t \in \target(\cS_t) \subseteq X(x)$.
        Define $F \colon X \to \potne{V}$ as $F(s) \ceq \target(\cS_{s(x)})$ --- notice that $s(x) \in \target(\cS_{s(x)})$
        holds by assumption, so $F$ is indeed well-defined.
        Then (\ref{claim: thetatarget is correct - (b)(i)}) is true, because for every $s \in X$ we have also assumed that $\target(\cS_{s(x)}) \subseteq X(x)$.
        Because $F(s)$ depends only on $s(x)$ and we have $X_t(x) = \{ t\}$, it follows that $X_t(x') = \target(\cS_t)$.
        As a result, we have $X_t(x) = \{ t \} \subseteq \target(\cS_t) = X_t(x')$ from which immediately follows that
        $X_t(x) \subseteq X_t(x') \in \target(\cG)$ for every $t \in X(x)$. Thus, by Claim \ref{claim: zeta unguarded},
        we obtain $\cG \models_{X_t} \zeta'(x,x')$ for every $t \in X(x)$, which is exactly (\ref{claim: thetatarget is correct - (b)(ii)}).
    \end{proof}
\end{exa}

\noindent
Now we have already one component, the $x$-myopic formula $\theta_\target(x)$, that is needed for the alternative proof of Theorem \ref{thm: myopic eso translates into myopic inex}.
However, we still require an interpretation lemma for $\inexlogic$.
Towards this end, consider a $\tau$-formula $\phi(x_1, \dotsc, x_m) \in \inexlogic$
and some first-order interpretation $\cI = (\delta, \epsilon, (\psi)_{S \in \tau})$.
The formula $\phi^\cI$ is defined as in section 2, whereas the occurring in-/exclusion atoms are handled as follows:
\begin{itemize}
    \item $(v_1,\dotsc,v_\ell \incl w_1,\dotsc,w_\ell)^\cI \ceq \E \tv_1'\dots \tv'_\ell \big( \bigwedge^\ell_{i=1} \delta(\tv'_i) \land \epsilon(\tv_i, \tv'_i) \land \tv_1,\dotsc,\tv_\ell \incl \tw_1,\dotsc,\tw_\ell\big)$
    \item $(v_1,\dotsc,v_\ell \excl w_1,\dotsc,w_\ell)^\cI \ceq
    \A \tv'_1\dotsc\tv'_\ell \big( \big[ \bigwedge^\ell_{i=1} \delta(\tv'_i) \land \epsilon(\tv_i, \tv'_i) \big] \ra \tv'_1,\dotsc,\tv'_\ell \excl \tw_1,\dotsc,\tw_\ell\big)$
\end{itemize}
Let $\fB \cong \cI (\fA)$ with coordinate map $h\colon \delta^\fA \to B$.
A team $X$ over $\fA$ with $\dom(X) = \{ \tx_1,\dots,\tx_m \}$
is said to be well-formed, if every $s \in X$ is well-formed (w.r.t.~$\cI$), i.e.~$s(\tx_i) \in \delta^\fA$ for $i=1,\dotsc,m$.
For such a team, $h(X) \ceq \set{h\circ s}{s \in X}$ is a well-defined team over $\fB$ with $\dom(h(X)) = \{x_1, \dots, x_m\}$.
The following lemma can be proven by induction over $\phi$.

\begin{lem}[Interpretation Lemma for $\inexlogic$]
    \label{lem: interpretation lemma for inex}
    For every well-formed team $X$ over $\fA$ with $\dom(X) = \{ \tx_1, \dots, \tx_m \}$, holds $\fA \models_X \phi^\cI \iff \fB \models_{h(X)} \phi$.
\end{lem}

\noindent
Consider any team $Y$ over $\fB$ with $\dom(Y) = \{ x_1, \dots, x_m \}$.
It is an easy consequence of this lemma, that $\fB \models_Y \phi \iff \fA \models_{h^{-1}(Y)} \phi^\cI$
where $h^{-1}(Y) \ceq \bigcup_{t \in Y} h^{-1}(t) = \set{s}{h \circ s \in Y}$ can be viewed as the ``full'' team describing $Y$.
Of course, different tuples of the base structure $\fA$ may encode the same element of the target structure $\fB$, thus $Y$ usually contains
redundant assignments.
For the same reason, two (well-formed) teams $X\neq X'$ (with the same domain and co-domain $\fA$) may describe the same team over $\fB$.
We say that $X$ and $X'$ are \emph{$h$-similar}, if $h(X) = h(X')$.

\begin{lem}[Similarity Lemma] % Wird (nur) für den alternativen Beweis benötigt.
    \label{lem: similarity lemma}
    Let $X, X'$ be well-formed teams that are $h$-similar. Then: $\fA \models_X \phi^\cI \iff \fA \models_{X'} \phi^\cI$.
\end{lem}

\begin{proof}
    Since $X,X'$ are $h$-similar, we have $h(X) = h(X')$.
    By the interpretation lemma, $\fA \models_X \phi^\cI \iff \fB \models_{h(X)} \phi \underset{h(X)=h(X')}\iff \fB \models_{h(X')} \phi \iff \fA \models_{X'} \phi^\cI$.
\end{proof}
\noindent
Now we have all the tools assembled that are needed to prove a stronger variant of Theorem \ref{thm: myopic inex - more general} resp.~Theorem \ref{thm: myopic eso translates into myopic inex}.

\begin{thm} \label{thm: union closed inex => myopic inex stronger}
    For every myopic \secorder-formula\footnote{It is well known that every $\inexlogic$-formula can be translated into an equivalent $\eso$-formula \cite{Gal12}. As we have already seen in Theorem \ref{thm: myopic <=> union-closed}, every union closed formula of existential \secorder logic is equivalent to some myopic $\eso$-formula. Therefore, we can w.l.o.g.~start with a myopic \secorder-formula.} $\phi(X)$ there exists an equivalent myopic formula $\mu(\tx)$ with only $8$ inclusion atoms and one exclusion atom.
\end{thm}

\begin{proof}
    Let $\phi(X)$ be a myopic $\eso$-formula.
    We are going to find a $\ty$-myopic formula $\mu(\ty) \in \inexlogic$ with $(\fA, X(\ty)) \models \phi(X) \iff \fA \models_X \mu(\ty)$
    for every $\tau$-structure $\fA$ and every team $X$ with $\ty \subseteq \dom(X)$.
    In the following let $\fA$ be a $\tau$-structure chosen arbitrarily.
    In Definition \ref{def: ugame for myopic Phi} we have defined the model-checking game $\cG \ceq \mc{\fA}{\phi}$ for $\fA$ and $\phi(X)$.
    By Proposition \ref{prop: myopic (A, X) |= Phi <=> X in I(G^A_Phi)}, we know for every \emph{relation} $X \subseteq A^r$ that:
    \begin{equation}
    (\fA, X)\models\phi(X) \iff X\in\target(\cG) \label{eq: relation <-> target}
    \end{equation}
    Due to Claim \ref{claim: thetatarget is correct} we may conclude for every \emph{team} $X$ with $x \in \dom(X)$:
    \begin{equation}
    \cG \models_X \theta_\target(x) \iff X(x) \in \target(\cG) \label{eq: team over game <-> target}
    \end{equation}
    By combining \eqref{eq: relation <-> target} and \eqref{eq: team over game <-> target}, we obtain that for every \emph{team} $X$ with $x\in \dom(X) $ over $\cG$:
    \begin{equation}
    \cG \models_X \theta_\target(x) \iff  (\fA, X(x))\models\phi(X) \label{eq: team <-> relation}
    \end{equation}
    Notice that $\cG \models_X \theta_\target(x)$ implies that $X(x) \subseteq T(\cG) = A^r$ and, hence, $X(x)$ is then actually a relation of the
    correct arity for the formula $\phi(X)$.
    By counting, it is easy to verify that exactly $8$ inclusion atoms and one exclusion atom occurs in $\theta_\target$.
    The formula $\mu(\tx)$ that we are going to build will result from $\theta_\target$ and use the same number of inclusion/exclusion atoms.

    Using the technique of Lemma \ref{lem: fo-interpretation for G_X}, it is possible to devise a (quantifier-free) interpretation
    $\cI = (\delta, \epsilon, \psi_{V_0}, \psi_{V_1}, \psi_{E}, \psi_{I}, \psi_{T}, \psi_{\Eexcl})$ such that $\cG \cong \cI(\fA)$ with coordinate map $h \colon \delta^\fA \ra V(\cG)$.
    This interpretation encodes a position of the game $\cG$ as a tuple $(\tu, \tv) \in A^{n+m}$ where the $n$-tuple
    $\tu$ has a certain equality type (indicating at which type of position we are encoding, e.g.~at which formula we are) while the $m$-tuple $\tv$
    stores certain values (e.g.~values of free variables and to which component the node belongs).
    More importantly, a position $\ta \in T(\cG) = A^r$ is described by the tuple $(\tu,\ta, \tb) \in A^{n+m}$ where $\tu$
    has equality type $e_T$ while $\tb \in A^{m-r}$ can be chosen arbitrarily.
    Also recall that every variable $v$ is replaced by an $(n+m)$-tuple $\tv$ of pairwise different variables.
    In particular, let $\tx = (\tu, \ty, \tz)$ where $\tu$ is a $n$-tuple, $\ty$ some $r$-tuple and $\tz$ an $(m-r)$-tuple.

    Every inclusion/exclusion atoms occurring in $\theta_\target$ has one of the following
    three possible forms where $v,w$ are some variables:
    \begin{itemize}
        \item $\beta_1(x,v,w) \ceq x v \excl x w$
        \item $\beta_2(x,v,w) \ceq xv \incl xw$
        \item $\beta_3(x,v) \ceq v \incl x$
    \end{itemize}
    Notice that the only inclusion atom of the form of $\beta_3$ is in $\theta_\target$ not within the scope of a disjunction (it is $x' \incl x$ right after
    the existential quantifier).
    In $\theta_\target^\cI$, these formulae are replaced by:
    \begin{itemize}
        \item $\beta^\cI_1(\tx,\tv,\tw) \ceq \A \tx' \A \tv' ([\delta(\tx') \land \delta(\tv') \land \epsilon(\tx, \tx') \land \epsilon(\tv, \tv')] \ra \tx' \tv' \excl \tx \tw)$
        \item $\beta^\cI_2(\tx,\tv,\tw) \ceq \E \tx' \E \tv' (\delta(\tx') \land \delta(\tv') \land \epsilon(\tx, \tx') \land \epsilon(\tv, \tv') \land \tx' \tv' \incl \tx \tw)$
        \item $\beta^\cI_3(\tx,\tv) \ceq \E \tv' (\delta(\tv') \land \epsilon(\tv, \tv') \land \tv' \incl \tx)$
    \end{itemize}
    Clearly, these formulas are not allowed in $\ty$-myopic formulae, because the occurring inclusion/exclusion atoms are not $\ty$-guarded.
    However, we know that the tuple $\tx = (\tu, \ty, \tz)$ is used in $\theta^\cI_\target$ to store elements from
    $h^{-1}(T(\cG)) = \set{(\tv, \ta, \tb) \in A^{n+m}}{\fA \models e_T(\tu), \ta \in A^r, \tb \in A^{m-r} }$,
    because whenever a team $Y$ interprets $\tx$ by values encoding different game positions, it follows
    that $h(Y)(x)$ cannot be a target set of a winning strategy, so
    Claim \ref{claim: thetatarget is correct} leads to $\cG \nmodels_{h(Y)} \theta_\target(x)$ and then, by the Interpretation Lemma (Lemma \ref{lem: interpretation lemma for inex}), $\fA \nmodels_Y \theta_\target^\cI$.
    So for every team $X$ which is well-formed (w.r.t.~$\cI$) and satisfies $\fA \models_X \theta^\cI_\target$ we must have that $\fA \models_X \psi_T(\tx) = e_T(\tu)$.
    This observation enables us to define versions of $\beta^\cI_i$ that are allowed in $\ty$-myopic formulae:
    \begin{itemize}
        \item $\beta^\star_1(\ty, \tv, \tw) \ceq \A \tv' ([\delta(\tv') \land \epsilon(\tv, \tv')] \ra \ty \tv' \excl \ty \tw)$
        \item $\beta^\star_2(\ty, \tv, \tw) \ceq \E \tv' (\delta(\tv') \land \epsilon(\tv, \tv') \land \ty \tv' \incl \ty \tw)$
        \item $\beta^\star_3(\ty, \tv) \ceq e_T(\tv_{(1)}) \land \tv_{(2)} \incl \ty$ where $\tv = \tv_{(1)} \tv_{(2)} \tv_{(3)}$
        and $|\tv_{(1)}| = n$, $|\tv_{(2)}| = r$ and $|\tv_{(3)}| = m-r$.
    \end{itemize}

    \begin{clm} \label{claim: beta replacing}
        For every team $X$ with $\dom(X) = \{ \tx, \tu, \tw \}$ with $X(\tx) \subseteq h^{-1}(T)$ and every $i=1,2,3$ holds
        $\fA \models_X \beta^\cI_i \iff \fA \models_X \beta^\star_i$.
    \end{clm}
    \begin{proof}
        $i = 1$:
        Let $X' \ceq X[\tx' \mapsto \delta^\fA, \tv' \mapsto \delta^\fA]\res{\epsilon(\tx, \tx') \land \epsilon(\tv, \tv')}$.
        Then we can observe that
        \[
            \fA \models_X \beta^\cI_1 \iff \fA \models_{X'} \tx' \tv' \excl \tx \tw \iff \fA \models_{X'} \tx' \tv' \excl \tx \tw \overset{(!)}\iff \fA \models_{X'} \ty \tv' \excl \ty \tw \iff \fA\models_{X} \beta_1^\star,
        \]
        but the equivalence of  $\fA \models_{X'} \tx' \tv' \excl \tx \tw$ and $\fA \models_{X'} \ty \tv' \excl \ty \tw$ requires proof.
        Let $\tx' \ceq (\tu', \ty', \tz')$ where $|\tu'|  = |\tu|, |\ty'|  = |\ty|$ and $|\tz'|  = |\tz|$.
        Because of $\fA \models_{X'} \epsilon(\tx, \tx')$ and $X'(\tx) \subseteq h^{-1}(T)$, we have that
        $s(\ty) = s(\ty')$ for every $s \in X'$.
        Now, we prove the two directions of $\fA \models_{X'} \tx' \tv' \excl \tx \tw \iff \fA \models_{X'} \ty \tv' \excl \ty \tw$ separately:

        ``$\Leftarrow$'':
        First, assume that $\fA \models_{X'} \ty \tv' \excl \ty \tw$.
        It follows that $s_1(\ty' \tv') = s_1(\ty \tv') \neq s_2(\ty \tw)$ for every $s_1,s_2 \in X'$.
        Because $\ty$ and $\ty'$ are subtuples of $\tx = (\tu, \ty, \tz)$ resp.~$\tx' = (\tu', \ty', \tz')$,
        this implies that $s_1(\tx' \tv') \neq s_2(\tx \tw)$ for every $s_1,s_2 \in X'$.
        Hence, $\fA \models_{X'} \tx' \tv' \excl \tx \tw$.

        ``$\Rightarrow$'':
        Now let $\fA \models_{X'} \tx' \tv' \excl \tx \tw$.
        Towards a contradiction assume that $\fA \nmodels_{X'} \ty \tv' \excl \ty \tw$.
        Then there are assignments $s_1, s_2 \in X'$ with $s_1(\ty \tv') = s_2(\ty \tw)$.
        Since $X(\tx) \subseteq h^{-1}(T)$ it follows that $\fA \models e_T(s_1(\tu)) \land e_T(s_2(\tu))$.
        Because we also have $\fA \models_{X'} \epsilon(\tx, \tx')$,
        it must be the case that $\fA \models_{s_1} e_T(\tu) \land e_T(\tu') \land \ty = \ty'$.
        Consider $s'_1 \ceq s_1[\tu' \mapsto s_2(\tu), \tz' \mapsto s_2(\tz)]$.
        $s_1$ and $s'_1$ only differ on $\tu'$ and $\tz'$, but both still encode the equality type $e_T$ in $\tu'$
        while the values of $\tz'$ are irrelevant.
        So we still have $\fA \models_{s'_1} \delta(\tx') \land \epsilon(\tx, \tx')$, implying that $s'_1 \in X'$.
        But then we have $s'_1(\tx') = s'_1(\tu', \ty', \tz') = (s'_1(\tu'), s'_1(\ty'), s'_1(\tz')) = (s_2(\tu), s'_1(\ty'), s_2(\tz))$
        and, because of $s'_1(\ty') = s'_1(\ty) = s_1(\ty) = s_2(\ty)$, we even have $(s_2(\tu), s'_1(\ty'), s_2(\tz)) = s_2(\tx)$
        proving that $s'_1(\tx') = s_2(\tx)$. Because we also have $s'_1(\tv') = s_2(\tw)$, this leads to
        $s'_1(\tx' \tv') = s_2(\tx \tw)$, which is impossible due to $\fA \models_{X'} \tx' \tv' \excl \tx \tw$. Contradiction!
        Therefore, $\fA \models_{X'} \ty \tv' \excl \ty \tw$ must be true.

        $i = 2$:
        ``$\Leftarrow$'':
        First let $\fA \models_X \beta_2^\star$. Then there is a function $F \colon X \to \potne{A^{n+m}}$ such that
        $\fA \models_{X'} \delta(\tv) \land \epsilon(\tv, \tv') \land \ty \tv' \incl \ty \tw$ for $X' \ceq X[\tv' \mapsto F]$.
        Due to $\fA \models_{X'} \ty \tv' \incl \ty \tw$, for every assignment $s \in X'$ there exists some $\tilde{s} \in X'$
        with $s(\ty \tv') = \tilde{s}(\ty \tw)$.
        Now consider $Y \ceq \set{s[\tx' \mapsto \tilde{s}(\tx)]}{s \in X'}$.
        Then we have that $\fA \models_Y \delta(\tx') \land \delta(\tv') \land \epsilon(\tx, \tx') \land \epsilon(\tv, \tv') \land \tx' \tv' \incl \tx \tw$ which proves that $\fA \models_X \beta^\cI_2$.

        ``$\Rightarrow$'':
        Now let $\fA \models_X \beta_2^\cI$.
        Then there is a function $F \colon X \to \potne{A^{2(n+m)}}$ such that
        $\fA \models_{X'} \delta(\tx') \land \delta(\tv') \land \epsilon(\tx, \tx') \land \epsilon(\tv, \tv') \land \tx' \tv' \incl \tx \tw$
        where $X' \ceq X[\tx'\tv' \mapsto F]$.
        Since $\fA \models_{X'} \tx'\tv' \incl \tx\tw$ it follows that $\fA \models_{X'} \ty'\tv' \incl \ty\tw$ (because $\ty, \ty'$ are subtuples of $\tx, \tx'$).
        Due to $X(\tx) \subseteq h^{-1}(T)$, $\fA \models_{X'} \delta(\tx') \land \epsilon(\tx, \tx')$ implies that
        $\fA \models_{X'} \ty = \ty'$. So, $\fA \models_{X'} \ty'\tv' \incl \ty\tw$ is equivalent to $\fA \models_{X'} \ty\tv' \incl \ty\tw$.
        Hence, we have that $\fA \models_{X'} \delta(\tv') \land \epsilon(\tv, \tv') \land \ty \tv' \incl \ty \tw$ which implies that
        $\fA \models_{X} \beta^\star_2$.

        $i = 3$:
        ``$\Leftarrow$'':
        First let $\fA \models_X \beta_3^\star$.
        Then $X(\tv_{(1)}) \subseteq e^\fA_T$ and $X(\tv_{(2)}) \subseteq X(\ty)$.
        So, for every $s \in X$ there exists some $\tilde{s} \in X$ with $s(\tv_{(2)}) = \tilde{s}(\ty)$.
        Consider $X' \ceq \set{s[\tv' \mapsto \tilde{s}(\tx)]}{s \in X}$.
        Since $\tilde{s}(\tx) \in  X(\tx) \subseteq h^{-1}(T)$, we have that $\fA \models_{X'} \delta(\tv') \land \epsilon(\tv, \tv') \land \tv' \incl \tx$
        which proves that $\fA \models_X \beta^\cI_3$.

        ``$\Rightarrow$'':
        Let $\fA \models_X \beta_3^\cI$.
        Then there is a function $F \colon X \to \potne{A^{n+m}}$ such that $\fA \models_{X'} \delta(\tv') \land \epsilon(\tv, \tv') \land \tv' \incl \tx$ where $X' \ceq X[\tv' \mapsto F]$.
        Let $\tv' = \tv'_{(1)}\tv'_{(2)} \tv'_{(3)}$ where $|\tv'_{(1)}| = n, |\tv'_{(2)}| = r $ and $|\tv'_{(3)}| = m-r$.
        Since $X'(\tx) = X(\tx) \subseteq h^{-1}(T)$, we have $X'(\tu) \subseteq e^\fA_T$ and, due to $\fA \models_{X'} \tv' \incl \tx$,
        it follows that $X'(\tv'_{(1)}) \subseteq e^\fA_T$.
        Furthermore, $\fA \models_{X'} \tv' \incl \tx$ also implies that $\fA \models_{X'} \tv'_{(2)} \incl \ty$.
        Because of $\fA \models_{X'} \epsilon(\tv, \tv')$ and $X'(\tv'_{(1)}) \subseteq e^\fA_T$,
        we also have $X'(\tv_{(1)}) \subseteq e^\fA_T$ and $\fA \models_{X'} \tv_{(2)} = \tv'_{(2)}$.
        This is why, we have that $\fA \models_{X'} e_T(\tv_{(1)}) \land \tv_{(2)} \incl \ty = \beta^\star_3(\ty, \tv)$.
        This concludes the proof of Claim \ref{claim: beta replacing}.
    \end{proof}
    Let $\theta^\star_\target(\tx)$ be the formula that results from $\theta^\cI_\target$ by replacing every subformulae $\beta^\cI_i(\tx, \dotsc)$
    by $\beta^\star_i(\ty, \dotsc)$.
    Now all the inclusion/exclusion atoms occurring in $\theta^\star_\target(\tx)$ go conform with the conditions of Definition \ref{def: myopic in/ex-formula} but it is still not quite a $\ty$-myopic formula, since $\theta^\star_\target(\tx) = \theta^\star_\target(\tu, \ty, \tz)$ has too many free variables.
    In order to get rid of the superfluous variables $\tu, \tz$ we simply define $\mu(\ty) \ceq \E \tu \E \tz (e_T(\tu) \land \theta^\star_\target(\tu, \ty, \tz))$ which is now a $\ty$-myopic formula.\footnote{Notice that the interpretation does \emph{not} introduce
        any disjunction above $\beta_3^\star(\tx, \tz)$, because $x \incl z$ is only in the scope of existential quantifiers and we only a need a conjunction in order to guard the translated existential quantifier in $\vartheta^\cI_\target$.}
    For a team $X$ with $\dom(X) = \{ \ty \}$ we call the team $Y$ a $T$-expansion of $X$, if
    $\dom(Y) = \{ \tu, \ty, \tz \}$,
    $Y(\ty) = X(\ty)$ and $Y(\tu) \subseteq e^\fA_T$.
    Clearly, we have that $\fA \models_X \mu(\ty)$ if and only if $\fA \models_Y \theta^\star_\target(\tx)$ for some $T$-expansion of $X$
    and two $T$-expansions of $Y,Y'$ are always $h$-similar to each other, because $h(Y(\tu, \ty, \tz)) = X(\ty) = h(Y'(\tu, \ty, \tz))$
    already implies that $h(Y) = h(Y')$ which, by the Similarity Lemma (Lemma \ref{lem: similarity lemma}),
    leads to $\fA \models_Y \theta_\target^\cI \iff \fA \models_{Y'} \theta_\target^\cI$.

    The full $T$-expansion $X_T$ of $X$ is
    $X_T \ceq \set{s[\tu \mapsto \ta, \tz \mapsto \tc]}{s \in X, \ta \in e^\fA_T, \tc \in A^{m-r}}$ while
    the ``game version'' $X_\cG$ of $X_T$ is a team with $\dom(X_\cG) = \{ x \}$ defined by $X_\cG \ceq \set{s_\tb}{\tb \in X(\ty)}$
    where $s_\tb \colon \{x\} \ra T(\cG),\, x \mapsto \tb$.
    It is not difficult to verify that $X_\cG(x) = X(\ty)$ and $h(X_T) = X_\cG$.

    Putting everything together, we obtain:
    \begingroup
    \allowdisplaybreaks
    \begin{align*}
        \fA \models_X \mu(\ty) & \iff \fA \models_Y \theta^\star_\target(\tx) \text{ for some $T$-expansion of $X$} \;\; \text{ (by construction of $\mu(\ty)$)} \\
        & \iff \fA \models_{Y} \theta^\cI_\target(\tx) \text{ for some $T$-expansion of $X$}   \;\; \text{ (follows from Claim \ref{claim: beta replacing})} \\
        & \iff \fA \models_{X_T} \theta^\cI_\target(\tx)    \;\; \text{ ($X_T$ is $h$-similar to every $T$-expansion of $X$)} \\
        & \iff \cG \models_{X_\cG} \theta_\target(x)        \;\; \text{ (Interpretation Lemma (Lemma \ref{lem: interpretation lemma for inex}))} \\
        & \iff (\fA, X_\cG(x)) \models \phi(X)              \;\; \text{ (due to \eqref{eq: team <-> relation})} \\
        & \iff (\fA, X(\ty)) \models \phi(X)                \;\; \text{ (because $X_\cG(x) = X(\ty)$)}
    \end{align*}
    \endgroup
    Thus, $\fA \models_X \mu(\ty) \iff (\fA, X(\ty)) \models \phi(X)$ follows and our proof is completed.
\end{proof}

\begin{cor}
    Every $\tx$-myopic formula $\mu(\tx)$ is equivalent to a $\tx$-myopic formula that uses exactly six guarded inclusion atoms,
    one inclusion atom of the form $\tx' \incl \tx$ and one guarded exclusion atom.
\end{cor}

\begin{proof}
    Every $\tx$-myopic formula $\mu(\tx)$ can be transformed into a equivalent myopic $\eso$-formula $\mu'(X)$
    and the construction used in the alternative proof produces exactly as many atoms as specified, if one replaces
    the unneeded subformula $\psiinit$ by any tautology (this does not harm, because union games do not have initial vertices, so $\psiinit$ was trivially satisfied).
\end{proof}

\subsection{Optimality of the Myopic Fragment of Inclusion-Exclusion Logic}

One might ask whether the restrictions of Definition \ref{def: myopic in/ex-formula} are actually imperative to capture the union closed fragment.
In this section, we will show that neither condition can be dropped and that every single atom of Definition \ref{def: myopic in/ex-formula}
is required to express all union closed properties.

We start by showing that neither condition can be dropped.
First of all, it is pretty clear that exclusion atoms have to be $\tx$-guarded, because
$x_1 \excl x_2$ is not guarded and obviously not closed under unions.
Furthermore, it is clear that the variables among $\tx$ must not be quantified.
This points out the necessity of conditions (\ref{def: myopic in/ex-formula - quantifiers}) and (\ref{def: myopic in/ex-formula - exclusion atoms})
of Definition \ref{def: myopic in/ex-formula}.
In the next example we demonstrate that neither restriction of condition (\ref{def: myopic in/ex-formula - inclusion atoms}) can be dropped.

\begin{figure}
    \begin{center}
        \begin{tikzpicture}[>=stealth, shorten >=0.25pt,align=center,minimum size=0, node distance=0cm, minimum height = 0.1 cm, minimum width = 0.1 cm]
        \newcommand{\length}{1.3}
        \newcommand{\hlength}{\length/2};
        \begin{scope}[shift={(-3.75,0)}]
        \newcommand{\labeldistance}{0.05}
        \newcommand{\f}{0.6};

        \tikzstyle{mynode} = [draw,circle,fill,inner sep=0,minimum size=7];

        % Vertices:
        \node[mynode] (a) at (-\hlength,0) {};
        \node[mynode] (b) at ($ (0,{sqrt(3*\hlength*\hlength)}) $){};
        \node[mynode] (c) at (+\hlength,0) {};

        \node[mynode] (aP1) at ($ {1 + \f}*(-\hlength,0) - \f*(+\hlength,0) $){};
        \node[mynode] (aP2) at ($ {1 + \f}*(-\hlength,0) - \f*(0,{sqrt(3*\hlength*\hlength)}) $){};
        \node[mynode] (bP) at ($ {1 + \f}*(0,{sqrt(3*\hlength*\hlength)}) - \f*(+\hlength,0) $){};
        \node[mynode] (bQ) at ($ {1 + \f}*(0,{sqrt(3*\hlength*\hlength)}) - \f*(-\hlength,0) $){};

        \node[mynode] (cQ1) at ($ {1 + \f}*(+\hlength,0) - \f*(-\hlength,0) $){};
        \node[mynode] (cQ2) at ($ {1 + \f}*(+\hlength,0) - \f*(0,{sqrt(3*\hlength*\hlength)}) $){};

        % Labels:
        % Vertex names:
        \node [right = \labeldistance of a] {$a$};
        \node [below = \labeldistance of b] {$b$};
        \node [left  = \labeldistance of c] {$c$};

        \node [left  = \labeldistance of aP1] {$a_+$};
        \node [below = \labeldistance of aP2] {$a_-$};
        \node [above = \labeldistance of bP] {$b_+$};
        \node [above = \labeldistance of bQ] {$b_-$};

        \node [right = \labeldistance of cQ1] {$c_+$};
        \node [below = \labeldistance of cQ2]  {$c_-$};

        % Predicates:
        \node [above = \labeldistance of aP1] {$P$};
        \node [left  = \labeldistance of aP2] {$P$};

        \node [left  = \labeldistance of bP] {$P$};
        \node [right = \labeldistance of bQ] {$Q$};

        \node [above = \labeldistance of cQ1] {$Q$};
        \node [right = \labeldistance of cQ2] {$Q$};

        % Edges:
        \draw (a) edge[->]  (b); % node[midway, above, sloped] {$E$}
        \draw (c) edge[->]  (b);
        \draw (b) edge[loop right] (b);

        \draw (a) edge[->, dashed] (aP1);
        \draw (a) edge[->, dashed] (aP2);

        \draw (b) edge[->, dashed] (bP);
        \draw (b) edge[->, dashed] (bQ);

        \draw (c) edge[->, dashed] (cQ1);
        \draw (c) edge[->, dashed] (cQ2);
        \end{scope}
        \begin{scope}[shift={(3.75,{sqrt(3*\hlength*\hlength)/2})}]
        \newcommand{\h}{{sqrt(\length * \length / 2)}}
        \newcommand{\bender}{15}
        \newcommand{\labeldistance}{0.05}
        \newcommand{\f}{0.6};

        \tikzstyle{mynode} = [draw,circle,fill,inner sep=0,minimum size=7];

        % Vertices:
        \node[mynode] (a) at (-2 * \length, 0) {};
        \node[mynode] (b) at (0, 0) {};
        \node[mynode] (c) at ( 2 * \length, 0) {};

        \node[mynode] (bP) at (-\length, 0) {};
        \node[mynode] (bM) at (+\length, 0) {};

        \node[mynode] (aP) at ($  (-2 * \length, 0) + (-\h, +\h)$) {};
        \node[mynode] (aM) at ($  (-2 * \length, 0) + (-\h, -\h)$) {};

        \node[mynode] (cP) at ($  ( 2 * \length, 0) + ( \h, +\h)$) {};
        \node[mynode] (cM) at ($  ( 2 * \length, 0) + ( \h, -\h)$) {};

        % Labels:
        % Vertex names:
        \node [above = \labeldistance of a] {$a$};
        \node [above = \labeldistance of b] {$b$};
        \node [above = \labeldistance of c] {$c$};

        \node [above = \labeldistance of aP] {$a_+$};
        \node [below = \labeldistance of aM] {$a_-$};

        \node [above = \labeldistance of bP] {$b_+$};
        \node [above = \labeldistance of bM] {$b_-$};

        \node [above = \labeldistance of cP] {$c_+$};
        \node [below = \labeldistance of cM] {$c_-$};

        % Edges:
        \draw (b)  edge[->, bend right = \bender] (bP);
        \draw (b)  edge[->, bend left  = \bender] (bM);

        \draw (bP) edge[->, bend right = \bender] (b);
        \draw (bM) edge[->, bend left  = \bender] (b);

        \draw (bP) edge[->] (a);
        \draw (bM) edge[->] (c);

        \draw (a)  edge[->] (aP);
        \draw (a)  edge[->] (aM);

        \draw (c)  edge[->] (cP);
        \draw (c)  edge[->] (cM);

        \draw (aM)  edge[->, loop left] (aM);
        \draw (aP)  edge[->, loop left] (aP);

        \draw (cM)  edge[->, loop right] (cM);
        \draw (cP)  edge[->, loop right] (cP);
        \end{scope}
        \end{tikzpicture}
    \end{center}
    \caption{The structures $\fA$ and $\fB$. The structure $\fA = (V, E^\fA, F^\fA, P^\fA, Q^\fA)$ on the left side uses
    two different kinds of edges: the dashed edges belong to $F$, while the other are $E$-edges.
    Furthermore, $\fA$ exhibits two predicates $P,Q$. The structure $\fB = (V,E^\fB)$ depicted on the right is just a directed graph.
    Please notice that both structures are using the same universe $V$.}
    \label{fig: structures for evil disjunctions}
\end{figure}
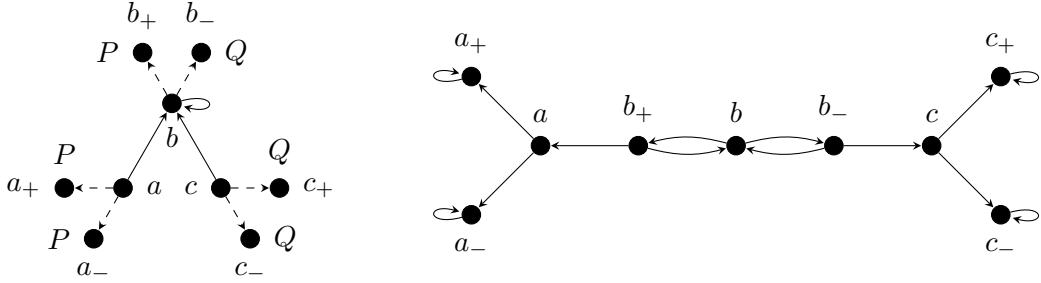

\begin{exa} \label{exam: evil disjunctions}
    Consider the structures $\fA$ and $\fB$ drawn in Figure \ref{fig: structures for evil disjunctions} and the following formulae:
    \begin{align*}
    \phi(x) \ceq\, & \E y \E z(Fxy \land Fxz \land xy \excl xz  \land [(Py \land \vartheta(x)) \lor (Qy \land \vartheta(x))]) \\
                   & \text{where }         \vartheta(x) \ceq \E v (Exv \land v \incl x) \\
    \psi(x) \ceq\, & \E y \E z(Exy \land Exz \land xy \excl xz  \land \E w (Eyw \land x \incl w) )
    \end{align*}
    Neither $\phi(x)$ nor $\psi(x)$ is $x$-myopic, because the inclusion atom $v \incl x$ from $\theta$
    occurs inside the scope of a disjunction (and it is not $x$-guarded), while the atom $x \incl w$ is neither $x$-guarded
    nor of the form that is allowed outside the scope of disjunctions, because $x$ appears on the wrong side of the inclusion atom.

    Let us take a moment to explain these formulae on an intuitive level.
    We will evaluate $\phi$ on $\fA$ in a team which assigns to $x$ values from $\{a,b,c\}$.
    There the formula states that one has to pick for every value of $x$ precisely
    one $F$-successor as a value for $y$, which is enforced by the exclusion atom.
    At $a$ and $c$ this choice is meaningless, because only the label ($P$ or $Q$) is relevant.
    Further, due to $\theta$, we must have made the same ``choice'' ($P$ or $Q$) for
    the vertex $b$, as we have made for the other vertices.
    In other words, the presence of $a$ compels to pick $b_+$ while the presence of $c$ forces $b_{-}$.
    Similarly, $\psi$, in the context of $\fB$, states that for every vertex
    among $\{a,b,c\}$ one successor can be chosen such that the neighbourhoods of the chosen successors cover the initial
    set of vertices.
    Notice that the relevant choice is made for the vertex $b$, which may move
    towards $a$, if $a$ is present in the team, or to $c$ otherwise.
    The other vertices found in the graphs are only there for technical reasons.

    For every $v \in V$ let $s_v \colon \{x\} \to V$ be the assignment with $s_v(x) \ceq v$.
    We define the teams $X_1 \ceq \{ s_a, s_b \}$, $X_2 \ceq \{ s_b, s_c \}$ and $X \ceq X_1 \cup X_2 = \{ s_a, s_b, s_c \}$.
    It is not difficult to verify that $\fA \models_{X_i} \phi(x)$ and $\fB \models_{X_i} \psi(x)$ for $i=1,2$ but
    $\fA \nmodels_X \phi(x)$ and $\fB \nmodels_X \psi(x)$.
    In particular, neither $\phi(x)$ nor $\psi(x)$ is closed under unions.
    This shows that the restrictions of Definition \ref{def: myopic in/ex-formula} are indeed necessary.
\end{exa}

\noindent
Thus the atoms allowed in Definition \ref{def: myopic in/ex-formula} are sufficient to capture the union closed fragment of $\inexlogic$.
On the contrary, one may ask whether the set of atoms given in Definition \ref{def: myopic in/ex-formula} is necessary.
Let us argue for all rules of Definition \ref{def: myopic in/ex-formula}.

Assume that all exclusion atoms are forbidden.
Then every formula is already in inclusion logic in which one cannot define every union closed property as was shown by Galliani and Hella \cite[p.~16]{GalHel13}.

If inclusion atoms were only allowed in the form $\tx\ty\incl\tx\tz$, that means the atoms $\ty\incl\tx$ are forbidden, the formulae become flat, as can be seen by considering Proposition \ref{prop: res inex components}, but not all union closed properties are flat.

The case where inclusion atoms of form $\tx\ty\incl\tx\tz$ are forbidden is a bit more delicate.
%The intuition here is that any such formula must be downwards closed for every given value of $\tx$.
To prove that such a formula cannot express every union closed property consider the formula $\mu(x) = \E z(z\incl x \wedge \A y(Exy\ra xy\incl xz))$, where $\tau=\{E\}$ for a binary predicate symbol $E$.
This formula axiomatises the set of all teams $X$ over a graph $G=(V,E)$ such that whenever $v\in X(x)$ and $(v,w) \in E$, then already $w \in X(x)$.
The formula obviously describes a union closed property.
Consider the graph $G$:
\tikz[baseline=-3]{\node[inner sep=1] (a) {$a$};\node[inner sep=1] (b) at (-1,0) {$b$}; \node[inner sep=1] (c) at (1,0) {$c$}; \draw (a) edge[->] (b) (a) edge[->] (c);}.
Here, $G\models_X\mu (x)$ for precisely those teams $X$ that satisfy ``$a\in X(x)$ implies $b,c\in X(x)$''.
For every $v \in V(G)$ let $s_v$ be the assignment $x \mapsto v$ and let $X_v \ceq \{ s_v \}$.
Furthermore, we define $X_{abc} \ceq \{s_a, s_b, s_c\}$.

Let $\psi(x)$ be an $x$-myopic formula in which the construct $x\ty\incl x\tz$ does not appear.
So the only inclusion atoms occurring in $\psi(x)$ are of the form $z \incl x$, which are not allowed
in the scope of disjunctions.
Notice that $z$ cannot be universally quantified, as the team $X_b=\{s_b\}$ satisfies the described property, but not $\A z(z\incl x)$.
Thus we may assume without loss of generality that $\psi(x)$ has the form $\E z(z\incl x\wedge \psi'(x,z))$, where in $\psi'(x,z)$ no atom of the kind $z'\incl x$ occurs.
We want to remark that the following argumentation can be adapted to the slightly more general case that multiple atoms of form $z\incl x$ occur, but for sake of simplicity we only deal with one such atom.
Let $\eta(x,z)$ be the unguarded version of $\psi'(x,z)$.
By Lemma \ref{lem: res inex components - normalform}, there is a function $F \colon X_{abc} \to \potne{V(G)}$ such
that $F(s) \subseteq X_{abc}(x) = V(G)$ for $s \in X_{abc}$ and $G \models_{X_{abc}[z \mapsto F] \res {x = v}} \eta$ for every $v \in X_{abc}(x)$.
Please notice that $X_{abc}[z \mapsto F] \res {x = v} = X_v[z \mapsto F]$.
Moreover, because in $\eta(x,z)$ no inclusion atom occurs it is downwards closed.
Assume $a\in F(s_a)$.
By downwards closure of $\eta(x,z)$ we obtain $G\models_{X_a[z\mapsto a]}\eta$, which, by Lemma \ref{lem: res inex components - normalform}, implies that $G\models_{X_a}\psi$ contradicting our assumption that $\psi$ describes the desired property.
Otherwise, because of symmetry, $b$ is in $F(s_a)$, and hence $G\models_{X_a[z\mapsto b ]}\eta$.
Additionally, since $G\models_{X_b}\psi$ we know, by Lemma \ref{lem: res inex components - normalform}, that $G\models_{X_b[z\mapsto b]}\eta$.
Together this implies $G\models_{X_{ab}[z\mapsto b] \res{x=v}}\eta$ for $v=a,b$ and, due to Lemma \ref{lem: res inex components - normalform}, we get $G\models_{X_{ab}}\psi$ which is again in conflict with our assumption about $\psi$ describing the desired property.

\section{An Atom capturing the Union Closed Fragment}
\label{sec: atom}

The present work was motivated by a question of Galliani and Hella in 2013 \cite{GalHel13}.
Galliani and Hella asked whether there is a union closed atomic dependency notion $\alpha$ that is definable in existential \secorder logic such that $\fo(\alpha)$ corresponds precisely to all union closed properties of $\inexlogic$.
In \cite{GalHel13} they have already shown that inclusion logic does not suffice, as there are union closed properties not definable in it.
Moreover, they have established a theorem stating that every union closed atomic property that is definable in first-order logic (where the formula has access to the team via a predicate) is expressible in inclusion logic.
Thus, whatever atom characterises all union closed properties of $\inexlogic$ must axiomatise an inherently \secorder property.

Intuitively speaking, as we have seen in Section \ref{sec:union games}, solving union games is a complete problem for the class $\union$.
Therefore, a canonical solution to this question is to propose an atomic formula that defines the winning regions in a union game.
Towards this we must describe how a game can be encoded into a team.
This is not as straightforward as one might think, because there is a technical pitfall we need to avoid.
The union of two teams describing union games, each won by player 0, might encode a game won by player 1, but by union closure it must satisfy the atomic formula.

We encode union games in teams by using variable tuples for the respective components, where we also encode the complementary relations in order to ensure that the union of two different games cannot form a different game.
For $k\in\bN$ let $\vark$ be the set of distinct $k$-tuples of variables $\{\tu,\tv_0,\tv_1,\tv,\tw,\tt,\tv_\mathsf{ex},\tw_\mathsf{ex},\tepsilon_1,\tepsilon_2,\tu^\complement,\tv^\complement,\tw^\complement,\tt^\complement,\tv^\complement_\mathsf{ex},\tw^\complement_\mathsf{ex},\tepsilon^\complement_1,\tepsilon^\complement_2\}$.

\begin{defi}
    \label{def:ugame inside team}
    Let $X$ be a team with $\vark\subseteq\dom(X)$ and codomain $A$.
    We define $\sim \ceq X(\tepsilon_1, \tepsilon_2)$ and $\fA^X \ceq (V, V_0, V_1, E, \init, T, \Eexcl)$ with the following components.
    \begin{multicols}{4}
        \begin{itemize}
            \item $V \ceq X(\tu)$
            \item $V_0 \ceq X(\tv_0)$
            \item $V_1 \ceq X(\tv_1)$
            \item $E \ceq X(\tv,\tw)$
            \item $\init \ceq \emptyset$
            \item $T \ceq X(\tt)$
            \item $\Eexcl\!\ceq X(\tv_\mathsf{ex},\tw_\mathsf{ex})$
        \end{itemize}
    \end{multicols}
    \noindent
    If the following consistency requirements are satisfied,
    then we define $\cG^A_X \ceq \fA^X\rsim$.
    \begin{enumerate}
        \item $X(\tu^\complement) = A^k \setminus V$ \label{def:ugame inside team(i)}
        \item $X(\tv^\complement,\tw^\complement) = (A^k\times A^k) \setminus E$ \label{def:ugame inside team(ii)}
        \item $X(\tt^\complement) = A^k \setminus T$
        \item $X(\tv_\mathsf{ex}^\complement,\tw_\mathsf{ex}^\complement) =( A^k\times A^k) \setminus \Eexcl$
        \item $X(\tepsilon_1^\complement, \tepsilon_2^\complement) = (A^k \times A^k) \setminus \sim$
        \item $V_0 = V\setminus V_1$ \label{def:ugame inside team(vi)}
        \item $\fA^X$ is a structure\footnote{This condition ensures that $E\subseteq V\times V$ and so forth.}. \label{def:ugame inside team(vii)}
        \item $\sim$ is a congruence on $\fA^X$.
        \item $\fA^X\rsim$ is a union game. \label{def:ugame inside team(ix)}
    \end{enumerate}
    \noindent
    Otherwise, if any of these requirements is not fulfilled, we let $\cG^A_X$ be undefined.
\end{defi}

\noindent
We call $X$ complete (w.r.t.~$A$), if $X(\ty) \cup X(\ty^\complement)$ is $A^k$ or $A^k\times A^k$ for every
$\ty \in  \{(\tu), (\tv, \tw), (\tt),$ $(\tv_{\mathsf{ex}}, \tw_{\mathsf{ex}}), (\tepsilon_1, \tepsilon_2) \}$ and $V = V_0 \cup V_1$,
and incomplete otherwise.
It is easy to observe that $\cG^A_X$ is undefined for every incomplete team $X$.
Furthermore complete subteams of teams describing a game actually describe the same game and the same congruence relation.

\begin{lem}
    \label{lem: complete subteams define the same game}
    Let $X,Y$ be teams with codomain $A$ and $\vark \subseteq \dom(X) = \dom(Y)$.
    If $X$ is complete, $X \subseteq Y$ and $\cG^A_Y$ is defined, then $\cG^A_X = \cG^A_Y$ and ${\sim_X} \ceq X(\tepsilon_1, \tepsilon_2) = Y(\tepsilon_1, \tepsilon_2) \eqc {\sim_Y}$.
\end{lem}
\begin{proof}
    Suppose that $X$ is complete, $X \subseteq Y$ and $\cG^A_Y$ is defined.
    First, we prove that $\cG^A_X$ is defined.
    Towards this end, we prove that $X$ satisfies the consistency requirements of Definition \ref{def:ugame inside team}.
    By completeness of $X$, we know already that $X(\tu) \cup X(\tu^\complement) = A^k$.
    Since $X \subseteq Y$ and $\cG^A_Y$ is defined, we have $X(\tu) \cap X(\tu^\complement) \subseteq Y(\tu) \cap Y(\tu^\complement) = \emptyset$.
    Thus, we have $X(\tu) \cup X(\tu^\complement) = A^k$ and $X(\tu) \cap X(\tu^\complement) = \emptyset$,
    which implies that $X(\tu^\complement) = A^k \setminus X(\tu)$.
    This proves condition (\ref{def:ugame inside team(i)}) of Definition \ref{def:ugame inside team}.
    The proof for (\ref{def:ugame inside team(ii)})-(\ref{def:ugame inside team(vi)}) is very analogous.

    Towards proving the remaining conditions (\ref{def:ugame inside team(vii)})-(\ref{def:ugame inside team(ix)}), it suffices
    to show that $\fA^X = \fA^Y$ and ${\sim_X} = {\sim_Y}$,
    because $\cG^A_Y$ is defined and thus the conditions (\ref{def:ugame inside team(vii)})-(\ref{def:ugame inside team(ix)}) must be true for $\fA^Y$ and $\sim_Y$.

    Thus, we need to prove that $X(\ty) = Y(\ty)$ for every tuple $\ty \in \{ (\tu), (\tv, \tw), (\tt), (\tv_{\cap}, \tw_{\cap}),$ $(\tepsilon_1, \tepsilon_2) \}$.
    Since the argumentation is very analogous for these different tuples, we prove this only for $\ty = \tu$.
    Towards a contradiction, assume that $X(\tu) \neq Y(\tu)$.
    Since $X$ and $Y$ are complete, we can conclude that $X(\tu^\complement) = A^k \setminus X(\tu)$ and $Y(\tu^\complement) = A^k \setminus Y(\tu)$.
    Since $X(\tu) \neq Y(\tu)$, we have that $Y(\tu) \setminus X(\tu) \neq \emptyset$ or $X(\tu) \setminus Y(\tu) \neq \emptyset$.
    In the first case, follows $\emptyset \neq Y(\tu) \setminus X(\tu) = Y(\tu) \cap (A^k \setminus X(\tu)) = Y(\tu) \cap X(\tu^\complement)
    \subseteq Y(\tu) \cap Y(\tu^\complement)$.
    In the second case, a similar line of thought leads to $\emptyset \neq X(\tu) \setminus Y(\tu) \subseteq Y(\ty) \cap Y(\tu^\complement)$.
    In both cases we have $Y(\tu) \cap Y(\tu^\complement) \neq \emptyset$  which contradicts $Y(\tu^\complement) = A^k \setminus Y(\tu)$.

    Therefore, all conditions of Definition \ref{def:ugame inside team} are fulfilled.
    Because of $\fA^X = \fA^Y$ and ${\sim_X} = {\sim_Y}$,
    it is even the case that $\cG^A_X = \fA^X_{/\sim_X} = \fA^Y_{/\sim_Y} = \cG^A_Y$.
\end{proof}

\noindent
Now let us show that union games are definable in plain first-order logic with team semantics in the sense of Definition \ref{def:ugame inside team}.

\begin{lem}
    \label{lem:ugames are definable in ts}
    Let $\phi(X)=\A\tx(X\tx\rightarrow\E\tR\phi'(X,\tR,\tx))$ be a myopic $\eso$-formula and $\psi(\vark, \tx)$ be a formula with team semantics (where
    $k$ is large enough such that the game $\mc{\fA}{\phi}$ can be encoded).
    There is a formula $\theta^\psi_\phi(\tx)$ such that $\fA\models_X\theta^\psi_\phi\iff\fA\models_Y\psi$ for some team $Y$ extending $X$ with $\cG^A_Y \cong \mc{\fA}{\phi}$ and $X(\tx)=Y(\tx)$, for every $\tau$-structure $\fA$.
\end{lem}

\begin{proof}
    Similar to Lemma \ref{lem: fo-interpretation for G_X},
    it is easy to construct a (quantifier-free) first-order interpretation $\cI \ceq (\delta, \epsilon, \psi_V, \psi_{V_0}, \psi_{V_1}, \psi_E, \psi_I, \psi_T, \psi_\Eexcl)$
    with $\cI(\fA) \cong \mc{\fA}{\phi}$.
    Now let $\theta^\psi_\phi(\tx) \ceq \A \vark (\gamma(\vark) \ra \psi(\vark, \tx))$ where the formula
    \begin{align*}
        \gamma(\vark) \;\ceq\;  & \delta(\tu) \land \psi_{V_0}(\tv_0) \land \psi_{V_1}(\tv_1)
        \land \psi_{E}(\tv, \tw) \land \psi_T(\tt) \land \psi_{\Eexcl}(\tv_{\mathsf{ex}}, \tw_{\mathsf{ex}}) \land  \epsilon(\tepsilon_1, \tepsilon_2) \,\land \\
        & \neg \delta(\tu^\complement) \land \neg \psi_{E}(\tv^\complement, \tw^\complement) \land \neg \psi_T(\tt^\complement) \land \neg \psi_{\Eexcl}(\tv^\complement_{\mathsf{ex}}, \tw^\complement_{\mathsf{ex}}) \land \neg \epsilon(\tepsilon^\complement_1, \tepsilon^\complement_2)
    \end{align*}
    enforces that the game $\mc{\fA}{\phi}$ will be ``loaded'' into the team.
    As long as none of these conjuncts are unsatisfiable this construction is correct.
    This is safe to assume because one can easily transform a union game into an equivalent one w.r.t.~the target set such that none of its components are empty.
\end{proof}
\noindent
This knowledge enables us to finally define the atomic formula we sought after.
For this we need to show that the atom is union closed and its first-order closure can express all of $\union$.

\begin{defi}
    \label{def: win atom}
    The atomic team formula $\win(\vark,\tx)$ for the respective tuples of variables has the following semantics.
    For non-empty teams $X$ with $\vark,\tx \subseteq \dom(X)$ we define
    \[\fA\models_X\win(\vark,\tx) \defiff X \text{ is complete and if  }\cG^A_X \text{ is defined, then }X(\tx)_{/X(\tepsilon_1, \tepsilon_2)}\in \target(\cG^A_X)\]
    and we set $\fA \models_\emptyset \win(\vark, \tx)$ to be always true (to ensure the empty team property).
\end{defi}

\noindent
Note that this atom can be defined in existential second-order logic.

\begin{prop}
    \label{prop: win atom union closed}
    The atomic formula $\win$ is union closed.
\end{prop}
\begin{proof}
    Assume that $\fA\models_{X_i}\win(\vark, \tx)$ for $i\in I$.
    We prove that $\fA \models_X \win(\vark, \tx)$ holds for the union $X \ceq \bigcup_{i \in I} X_i$.
    If $X = \emptyset$, there is nothing to prove.
    Otherwise at least one $X_j$ is non-empty and, since $\fA \models_{X_j} \win(\vark, \tx)$, $X_j$
    must be complete implying that $X$ is also complete (because $X \supseteq X_j$).
    For the remainder of this proof, we assume w.l.o.g.~that all involved teams $X_i$ (and $X$) are non-empty.
    If $\cG^A_X$ is undefined, then $\fA\models_X \win(\vark, \tx)$ follows from the definition of $\win$.
    Otherwise, if $\cG^A_X$ is defined, then we can use Lemma \ref{lem: complete subteams define the same game}
    to obtain that $\cG^A_{X} = \cG^A_{X_i}$ and ${\sim} \ceq X(\tepsilon_1, \tepsilon_2) = X_i(\tepsilon_1, \tepsilon_2)$ for every $i \in I$.
    Since $\fA \models_{X_i} \win(\vark, \tx)$, we can conclude that $X_i(\tx)\rsim \in \target(\cG^A_{X_i}) = \target(\cG^A_X)$
    for each $i \in I$.
    By Theorem \ref{thm: ugames are union closed}, $X(\tx)\rsim = \bigcup_{i \in I} X_i(\tx)\rsim \in \target(\cG^A_X)$
    and, hence, $\fA \models_X \win(\vark, \tx)$.
\end{proof}

\begin{thm}
    \label{thm: win can express all union closed properties of sigma11}
    Let $\phi\in\inexlogic$ be a union closed formula.
    There is a logically equivalent formula $\zeta\in\fo(\win)$.
    In other words, $\fo(\win)$ corresponds precisely to the union closed fragment of $\inexlogic$.

    \begin{proof}
        Let $\fA$ be an arbitrary structure.
        Due to \cite[Theorem 6.1]{Gal12} there exists a formula $\phi'(X)\in\eso$ which is logically equivalent to $\phi(\tx)$ in the sense that $\fA\models_X\phi(\tx)\iff(\fA,X(\tx))\models\phi'(X)$ for every team $X$ with $\tx \subseteq \dom(X)$.
        By Theorem \ref{thm: myopic <=> union-closed}, there is a myopic formula $\mu \equiv \phi'$.
        So, we have $(\fA,X(\tx)) \models \mu(X) \iff \fA \models_X \phi(\tx)$.

        The game $\mc{\fA}{\mu}$ from Definition \ref{def: ugame for myopic Phi} is a union game
        and Lemma \ref{lem:ugames are definable in ts} allows us to load this game into a team.
        Please notice, that Lemma \ref{lem:ugames are definable in ts} is using a similar first-order interpretation $\cI$ as Lemma \ref{lem: fo-interpretation for G_X},
        which encodes a target vertex $\ta \in T(\mc{\fA}{\mu})$
        by tuples of the form $(\tu, \ta, \tw)$ of length $k= n+m$ where the $n$-tuple $\tu$ has the equality type $e_T$
        while $\tw$ is an arbitrary tuple of length $m-|\ta|$.
        Let $\psi(\vark, \tx) \ceq \A \tu \A \tw (e_T(\tu) \ra \win(\vark,\tu \tx \tw))$ and $\zeta(\tx)\ceq\theta^{\psi}_\mu$ be as in Lemma \ref{lem:ugames are definable in ts}, that is $\A \vark(\gamma(\vark)\ra \psi(\vark,\tx))$.
        So $\fA \models_X \zeta(\tx) \iff \fA \models_Y \psi(\vark, \tx)$
        where $Y = X[\vark \mapsto A]\res{\gamma}$.
        As in Lemma \ref{lem:ugames are definable in ts},
        we have $\cG^A_Y \cong \cI(\fA) \cong \cG(\fA, \mu)$ and $X(\tx) = Y(\tx)$.
        Furthermore, we have defined $\cG^A_Y = \fA^Y\rsim$ where ${\sim} \ceq Y(\tepsilon_1, \tepsilon_2)$.

        Because of the construction of $\psi$, we have
        $\fA \models_Y \psi(\vark, \tx) \iff \fA \models_Z \win(\vark,\tu\tx\tw)$ where $Z \ceq Y[\tu \mapsto e^\fA_T, \tw \mapsto A^{m-|\tx|}]$.
        Since $\cG^A_Z = \cG^A_Y \cong \mc{\fA}{\mu}$ is a well-defined union game, this is equivalent to
        $Z(\tu \tx \tw)\rsim \in \target(\cG^A_Y)$.
        Let $h\colon \delta^\fA_\cI \to V(\mc{\fA}{\mu})$ be the coordinate map for $\mc{\fA}{\mu} \cong \cI(\fA)$.
        By construction, $h$ induces an isomorphism between $\fA^Y\rsim$
        and $\mc{\fA}{\mu}$.
        In particular each element of any equivalence class $[(\tu', \ta, \tw')]_{\sim} \in Z(\tu \tx \tw)\rsim$ is mapped by $h$ to $\ta$.
        Therefore, $Z(\tu \tx \tw)\rsim \in \target(\cG^A_Y) \iff Z(\tx) = X(\tx) \in \target(\mc{\fA}{\mu})$.
        Thus we have $\fA \models_X \zeta(\tx) \iff X(\tx) \in \target(\mc{\fA}{\mu})$.
        Putting everything together, we have
        $\fA \models_X \zeta(\tx) \iff X(\tx) \in \target(\mc{\fA}{\mu}) \iff (\fA,X(\tx)) \models \mu \iff \fA \models_X \phi(\tx)$
        as desired.
    \end{proof}
\end{thm}

\section{Inclusion and Exclusion Games}\label{sec: in and ex games}

In this section we provide evidence that inclusion-exclusion games can be specialised to fit different fragments of $\eso$ or, equivalently, $\inexlogic$.
We present inclusion as well as exclusion games as natural syntactical restrictions of inclusion-exclusion games\footnote{This should also clarify the name choice ``inclusion-exclusion game'', apart from being games for $\inexlogic$.}.
Recall that in an inclusion-exclusion game $\cG$ the inclusion edges $\Eincl$ are defined as $E \cap V\times T$, that is the set of all edges going \emph{into} any terminal vertices $T$.
To be more specific what we mean by naturalness, exclusion games are inclusion-exclusion games without inclusion edges and inclusion games are inclusion-exclusion games without exclusion edges.
Furthermore, as we shall see, these games capture exclusion logic $\exclogic$ and inclusion logic $\inclogic$, respectively.

\subsection*{Exclusion Games}

\begin{defi}
    \label{def: ex game}
    An \emph{exclusion game} is an inclusion-exclusion game $\cG = (V, V_0, V_1, E, \init, T, \Eexcl)$ subject to $\init = \emptyset$ and $\Eincl = \emptyset$.
\end{defi}

\noindent
Throughout this section $\cG$ refers to an exclusion game.
Let us remark on the stipulation $\init = \emptyset$.
Since exclusion games are intended as model-checking games for $\exclogic$, a key property must be that their target set is downwards closed, which cannot be guaranteed if initial vertices are present.

\begin{prop}
    \label{prop: ex game target set}
    The set $\target(\cG)$ is downwards closed.
\end{prop}
\begin{proof}
    Let $B \in \target(\cG)$ and $A\subseteq B$.
    Put $C \ceq B\setminus A$.
    By assumption there is a winning strategy $\cS$ with $\target(\cS) = B$.
    Since $\init = \emptyset$ and $\Eincl = \emptyset$ it follows at once that $\cS - C$ is a winning strategy with $\target(\cS) = A$.
\end{proof}

\newcommand{\exgame}[2]{\cG_{\mathsf{ex}}(#1, #2)}

\begin{defi}
    \label{def: ex game for ex logic}
    Let $\phi(\tx) \in \exclogic$ and $\fA$ be a structure.
    The exclusion game $\exgame{\fA}{\phi} = (V, V_0, V_1, E, \init = \emptyset, T, \Eexcl)$ is defined as follows.
    \begin{itemize}
        \item $V \ceq \{(\psi, s) : \psi\in\subf(\phi), s\colon \free(\psi) \to A\} \cup T$
        \item $T \ceq \{s : s\colon \{\tx\}\to A\}$
        \item $V_0 \ceq \{(\theta, s) : \theta = \E x\,\gamma, \theta = \gamma\lor\gamma'\text{ or $\gamma$ is a literal and }\fA\nmodels_s\gamma\}$
        \item $V_1 \ceq V \setminus V_0$
        \item $\begin{aligned}[t]
        E \,\ceq\, & \set{((\gamma \circ \theta,s),(\delta,s\res{\free(\delta)}))}{ \circ \in \{ \land, \lor \}, \delta \in  \{ \gamma, \theta \} } \, \cup \\
        & \set{((X\tx,s),s(\tx))}{X\tx \in \subf(\phi')} \, \cup  \\
        & \set{((Qx\gamma,s), (\gamma, s'))}{Q \in \{ \E, \A \}, s' = s[x \mapsto a], a \in A } \, \cup \\
        & \set{(s,(\phi,s))}{s \in T },
        \end{aligned}$
%        \item $\begin{aligned}[t]
%        		E \ceq &\set{(s,(\phi, s))}{s \colon \{ \tx \} \to A \}} \cup \\
%        		       &\{((\theta, s), (\theta', s')) : \theta' \text{ is a direct subformula of } \theta, \theta\text{ is not a literal, and } s(v) = s'(v) \text{ for all } v \in \free(\theta) \cap \free(\theta')\}
%        	   \end{aligned}$
        \item $\Eexcl \ceq \{((\tv\excl\tw, s), (\tv\excl\tw, s')) : s(\tv) = s'(\tw)\}$
    \end{itemize}
\end{defi}

\noindent
It is not difficult to see the resemblance of this definition with the semantics of an exclusion logic formula, hence the following theorem is an immediate observation, whose proof can be executed by a standard induction.

\begin{thm}
    \label{thm: ex game = ex logic}
    Let $\fA$ be a structure, $\phi \in \exclogic$ and $\cG \ceq \exgame{\fA}{\phi}$.
    For all teams $X$ with $\dom(X) = \{ \tx \}$ holds $\fA \models_X \phi \iff X \in \target(\cG)$.
\end{thm}

\subsection*{Inclusion Games}

\begin{defi}
    \label{def: in game}
    An \emph{inclusion game} is an inclusion-exclusion game $\cG = (V, V_0, V_1, E, \init, T, \Eexcl)$ subject to $\init = \emptyset$ and $\Eexcl = \emptyset$.
\end{defi}

\noindent
The target set of an inclusion game is closed under unions, because the union of winning strategies inherits properties (\ref{def: winning strategy - player 0}) through (\ref{def: winning strategy - initial position}) of Definition \ref{def: winning strategy} and because exclusion edges are forbidden condition (\ref{def: winning strategy - exclusion edges}) trivialises.
Moreover, every inclusion game $\cG_\text{inc}$ can be turned into an equivalent union game by constructing the model-checking game $\mc{\cG_\text{inc}}{\phitarget(X)}$
according to Definition \ref{def: ugame for myopic Phi} where $\phitarget(X)$ is the myopic sentence from the proof of Theorem \ref{thm: union closed eso => myopic stronger}.
These games are indeed equivalent, because the combination of Claim \ref{claim: phitarget is correct} and Proposition \ref{prop: myopic (A, X) |= Phi <=> X in I(G^A_Phi)} immediately yields $\target(\cG_\text{inc}) = \target(\mc{\cG_\text{inc}}{\phitarget(X)})$.

\erich investigated, amongst other things, model-checking games for inclusion- and greatest fixed-point logic in \cite{Graedel16}.
He used the notion of \emph{second-order reachability games} which are similar in nature to inclusion-exclusion games and have in fact inspired these.
Instead of using target vertices \erich coined the notion of \emph{$I$-traps} that are subsets $X$ of the initial vertices $I$ such that player 0 has a second-order strategy playing from $X$ that never reaches any vertex in $I\setminus X$ (in this sense it is a trap for the opponent regarding only the initial vertices).
Because the winning condition of these games requires to secure that the play remains inside a region they are called \emph{safety games}.
We repeat this notion with minor adjustments to fit our vocabulary.

\begin{defi}
    \label{def: safety game}
    A \emph{safety game} $G = (V, V_0, V_1, I ,E)$ is a tuple of vertices $V = V_0 \cupdot V_1$ divided into positions of player 0 and player 1, respectively, initial vertices $I$ and a set of possible moves $E\subseteq V\times V$.

    An \emph{$I$-trap} is a set $X \subseteq I$ such that player 0 has a winning strategy, i.e.~a subgraph $\cS = (W, F)$ of $(V, E)$, with $I \cap W = X$.
    For $(W, F)$ to be winning the following conditions have to be satisfied.
    \begin{enumerate}
        \item For all $v\in V_0\cap W$ holds $\nh{\cS}{v}\neq\emptyset$.
        \item For all $v\in V_1\cap W$ holds $\nh{\cS}{v} = \nh{G}{v}$.
    \end{enumerate}
\end{defi}

\noindent
Further, \erich associated with every inclusion logic formula $\psi$ and structure $\fA$ a safety game $\cG_{\mathsf{safe}}(\fA, \psi)$ and proved in the following precise sense that it is a model-checking game for $\psi$ which resembles Theorem \ref{thm: ex game = ex logic}.

\begin{thmC}[\cite{Graedel16}]
    \label{thm: erich safety games}
    For all structures $\fA$, $\psi\in\inclogic$ and teams $X$ we have
    \[
        \fA \models_X \psi \iff \psi\times X \text{ is an $I$-trap in } \cG_{\mathsf{safe}}(\fA, \psi).
    \]
\end{thmC}

\noindent
Because inclusion games neither have initial vertices nor exclusion edges conditions (\ref{def: winning strategy - initial position}) and (\ref{def: winning strategy - exclusion edges}) of Definition \ref{def: winning strategy} vanish.
What remains are only those stipulations also required for a winning strategy associated with an $I$-trap.
Moreover, an $I$-trap $X$ essentially corresponds to a winning strategy with target set $X$.
Hence we observe that every safety game can be viewed as an inclusion game and vice versa.

\begin{obs}
    \label{obs: safety games = inclusion games}
    For every inclusion game $\cG$ there is a safety game $\cG_{\mathsf{safe}}$ such that for all winning strategies $\cS$ there is an $I$-trap $\target(\cS)$ in $\cG_{\mathsf{safe}}$ and vice versa.
\end{obs}

\noindent
It follows at once that inclusion games are model-checking games for inclusion logical formulae.
%Galliani and Hella established a connection between inclusion logic and greatest fixed-point logic \cite{GalHel13} together with the Immermann-Vardi Theorem \cite{immermann, vardi} that least fixed-point logic (also greatest fixed-point logic) captures $\ptime$ on the class of all ordered structures we obtain

\section{Concluding Remarks}

The central notion in this paper are ``inclusion-exclusion games'' which, as we have proven are model-checking games for existential second-order logic $\eso$, or, equivalently, inclusion-exclusion logic $\inexlogic$.
A syntactical restriction of these games lead to union games, the model-checking games for the class of all formuale of $\eso$ that are closed under unions, class defined by its semantics.
In Section \ref{sec: in and ex games} other restrictions gave access to games for exclusion and inclusion logic, respectively.
We want to point out that even for other fragments, such as closure under intersection, negation, super sets and so forth, it is not difficult to find variations of inclusion-exclusion games that capture those.
We refrain from making the constructions explicit as it is a simple exercise and we do not see how these special cases are of any benefit.
However, the mere fact that inclusion-exclusion games are easily adapted to fit any such setting makes us belief that they are a natural notion.

Let us remark on the ``naturalness'' of the atom $\win$.
Certainly inclusion, exclusion and the notions alike can be regarded as natural atomic dependency formulae, whereas the just introduced atom has to be classified differently.
Nevertheless, it is a canonical candidate since it solves a complete problem of the desired class.
Of course, a more natural~--- and more usable~--- atom might be found, but it will not be as simplistic as e.g.~inclusion for Galliani and Hella have shown that every first-order definable union closed property is already expressible in inclusion logic.
Hence, whatever atom one proposes, it must make use of some inherently second-order concepts.
For concretely expressing properties, the introduced myopic fragments of $\eso$ and $\inexlogic$ are more practical.

The various syntactical characterisations of the union closed fragments presented in this work now enables their further investigation.
This could result in a complexity theoretical analysis or a more detailed classification of $\eso$.

\subsection*{Acknowledgements}

We thank the anonymous referees for their helpful comments.

\bibliographystyle{alpha}
\bibliography{unionLMCS}

\end{document}